\documentclass[10pt, journal]{IEEEtran}
\usepackage{cite}
\ifCLASSINFOpdf
   \usepackage[pdftex]{graphicx}
    \usepackage{subfigure}
\else
\fi
\usepackage{amsmath}
\usepackage{amssymb}
\usepackage{amsthm}
\usepackage{bm}
\usepackage{color}

\usepackage{enumerate}

 \usepackage{subfigure}
\usepackage{float}

\usepackage{diagbox}

\newtheorem{remark}{Remark}

\newtheorem{coro}{\textbf{Corollary}}

\newtheorem{theorem}{\textbf{Theorem}}

\newcommand{\expect}[1]{\mathbb{E}\big\{#1\big\}}

\DeclareMathOperator*{\argmax}{arg\,\max}
\DeclareMathOperator*{\argmin}{arg\,\min}
\DeclareMathOperator*{\Ex}{\mathbb{E}}

\allowdisplaybreaks
\begin{document}

\title{Timely-Throughput Optimal Scheduling with Prediction}

\author{\IEEEauthorblockN{Kun Chen and Longbo Huang}\\
\IEEEauthorblockA{IIIS, Tsinghua University\\ 
% Atlanta, Georgia 30332--0250\\
chenkun14@mails.tsinghua.edu.cn, longbohuang@tsinghua.edu.cn}
%\and
%\IEEEauthorblockN{Longbo Huang}
%\IEEEauthorblockA{IIIS, Tsinghua University\\
%longbohuang@tsinghua.edu.cn}
\thanks{Kun Chen and Longbo Huang are with the Institute for Interdisciplinary Information Sciences. 
This work was supported in part by the National Natural Science Foundation of China Grants 61672316, 61303195, the Tsinghua Initiative Research Grant, and the China Youth 1000-talent grant.}
}


%Telephone: (800) 555--1212, Fax: (888) 555--1212}


\maketitle 

\begin{abstract} 
Motivated by the increasing importance of providing delay-guaranteed services in general computing and communication systems, and the recent wide adoption of learning and prediction in network control, in this work, we consider a general stochastic single-server multi-user system and investigate the fundamental benefit of predictive scheduling in improving timely-throughput, being the rate of packets that are delivered to destinations before their deadlines. 
%where each packet has a hard deadline, and investigate the fundamental benefit of predictive scheduling in improving timely-throughput. 
%such a system. The packets that are not successfully delivered to their destinations within their deadlines are discarded from the system and are not counted in the throughput of the corresponding user. 
% 
% 
By adopting an error rate-based prediction model, we first derive a Markov decision process (MDP) solution to optimize the timely-throughput objective subject to an average resource consumption constraint. 
Based on a packet-level decomposition of the MDP, we explicitly characterize the optimal scheduling policy and rigorously quantify the timely-throughput improvement due to predictive-service, which scales as $\Theta(p\left[C_{1}\frac{(a-a_{\max}q)}{p-q}\rho^{\tau}+C_{2}(1-\frac{1}{p})\right](1-\rho^{D}))$, where $a, a_{\max}, \rho\in(0, 1), C_1>0, C_2\ge0$ are constants, $p$ is the true-positive rate in prediction, $q$ is the false-negative rate, $\tau$ is the packet deadline and $D$ is the prediction window size. 
We also conduct extensive simulations to validate our theoretical findings. 
Our results provide novel insights into how prediction and system parameters impact  performance and provide useful guidelines for designing  predictive low-latency control algorithms.

%relationship of optimal scheduling in the prediction case, perfect or imperfect, with that in the no prediction case, and provide an approach to find optimal policies that can maximize any specified weighted total timely throughput under average energy constraint. 

%We explicitly give the optimal policy and the throughput improvement due to prediction in a deterministic channel state scenario, and characterize the optimal policy in the general scenario. Simulation results are presented to validate our results.
\end{abstract} 


% \fi
\IEEEpeerreviewmaketitle

\section{Introduction} 
How to provide low-latency packet delivery has long been an important problem in  network optimization research, particularly due to the increasingly more stringent user delay requirements in a wide range of applications. For instance, low delay is critical for video traffic in mobile networks, which has already accounted for $60\%$ of total mobile data in $2016$ and will account for more than $78\%$ by $2021$ according to a recent Cisco report \cite{cisco-17}. Other areas such as online gaming, online health care and supply chain also have rigid delay requirements. 
Indeed, user requirements are so strong, that it has been reported that for companies like Amazon and Google, if their service latency increases by $500$ms, they will lose $1.2\%$ of their customers and millions of dollars revenue \cite{delay-report-09}. 
As a result, the problem of guaranteeing low-latency  has received much attention in the last decade, and many scheduling algorithms have been designed based on various mathematical techniques, e.g., \cite{tassiulas93}, \cite{neelysuperfast}, \cite{buisrikant_infocom09}, \cite{ying_wmshortest_infocom09}, \cite{huangneely_dr_tac}, \cite{hou-delay-10}, \cite{singh2016throughput}. 

On the other hand, driven by the availability of large amount user behavior data and the rapid development of data mining and machine learning tools, it has become common in practical systems to \emph{predict} user demand and to \emph{proactively} serve customer requests. 
For example, Amazon tries to predict what customers may purchase and pre-ships products to distribution centers close to them, in order to reduce shipping time\cite{amazon-pre-ship-14}. 
Netflix, on the other hand, tries to predict what customers may want and preload videos onto user devices to improve quality-of-experience \cite{netflix-16}.  
Another example is brunch prediction in computer architecture, where prediction is used to decide how to pre-execute certain parts of the workload, so as to reduce computing time \cite{branch-predict}.  
Despite the continuing success of this prediction-based approach in practice,  it has not received much attention in theoretical study. Therefore, it remains largely unknown how prediction can fundamentally improve delay-guaranteed services.  

In this paper, we aim to fill this gap and investigate \emph{the impact of prediction on  timely-throughput}. Specifically, we consider a single-server multi-user system where the server delivers  packets to users. Each packet has a user-dependent deadline before which it needs to reach the user. %Otherwise, it becomes useless and will be discarded. 
The service channel for each user is time-varying and the transmission success probability depends on the resource spent sending a packet. 
% 
%The server gets access to an \emph{imperfect} prediction window, which forecasts future arrivals, and can \emph{pre-serve} packets before they actually enter the system. 
The server gets access to an \emph{imperfect} prediction window, in which forecasts about future arrivals are available, and can \emph{pre-serve} packets before they actually enter the system. 
The overall objective of the system is to maximize a weighted sum of timely-throughputs of users, being the rates of packets delivered before their deadlines. 
This formulation is general and models various important practical applications, e.g., video streaming, sending time-critical control information, and grocery delivery.

There has been an increasing set of recent results investigating the impact of prediction in networked system control.  \cite{huang2016backpressure} and \cite{haoran-predictive-16} consider utility optimal scheduling in downlink systems based on perfect user prediction.  \cite{zhang-predict-ton} shows that proactive scheduling can effectively reduce queueing delay in stochastic single-queue systems.  \cite{plc-mobihoc-17} considers how network state prediction can be incorporated into algorithm design.  \cite{proactive-caching-uncertainties} and \cite{proactive-download-shaping-ton15} focus on understanding the cost saving aspect of proactive scheduling based on demand prediction.   \cite{ooc-sig-15}, \cite{ooc-sig-16} and \cite{rchase-eenergy16} also investigate the benefit of prediction from an online algorithm design perspective. However, we note that the aforementioned works all focus on understanding the utility improvement aspect of prediction and proactive service, and delay saving often comes as a by-product of the resulting predictive control algorithms. Thus, the results are not applicable to delay-constrained problems, where meeting the latency guarantee is an explicit requirement. 

Our formulation is closest to recent works \cite{hou-delay-10}, \cite{hou-rtss}, \cite{huang2016backpressure}, and  \cite{singh2016throughput}, which focus on delay-constrained traffic scheduling. Our work is different as follows. \cite{hou-delay-10} focuses on the setting where traffic is generated and delivered within synchronized frames for all users and \cite{hou-rtss} focuses on periodic traffic, while our work allows heterogeneous deadlines for user packets and random arrivals. \cite{huang2016backpressure} focuses on optimizing system utility subject to stability constraint, while we work explicitly with delay constraints. 
Lastly, while our work builds upon the novel results in \cite{singh2016throughput}, we focus on quantifying the impact of prediction and proactive service in a Markov system, whereas \cite{singh2016throughput} considers a causal system with an i.i.d. setting.  
The extension to incorporate prediction significantly complicates both the solution and analysis. Our results also offer novel insights into the fundamental benefits of prediction in delay-constrained network control. 

The main contributions of our paper are summarized as follows. %\textcolor{red}{= needs revision =}

(i) We propose a novel framework for studying timely-throughput optimization with imperfect prediction and proactive scheduling. Our model captures key features of practical delay-constrained problems and facilitates analysis. 

(ii) We derive the exact optimal solution to the prediction-based timely-throughput optimization problem using Markov decision process (MDP). We  rigorously quantify that prediction improves timely-throughput by $\Theta(p\left[C_{1}\frac{(a-a_{\max}q)}{p-q}\rho^{\tau}+C_{2}(1-\frac{1}{p})\right](1-\rho^{D}))$, where $a, a_{\max}, C_1>0, C_2\ge0, \rho\in(0, 1)$ are constants, $p$ is the true-positive rate in prediction, $q$ is the false-negative rate, $\tau$ is the packet deadline and $D$ is the prediction window size. This concise and explicit characterization is novel and provides insights into how different parameters impact system performance. 

(iii) We conduct extensive simulations to validate our theoretical findings. Our results show that prediction-based system control can significantly boost timely-throughput. %\textcolor{red}{= compared to causal approaches? boost resource consumption?=} 

The rest of the paper is organized as follows. In Section \ref{section:model}, we present the system model. The MDP-based solution is presented in Section \ref{section:mdp}. Structural properties of the optimal solution and exact timely-throughput improvement for a static setting are derived in Section \ref{sec:simple-case}. The general scenario is considered in Section \ref{sec:general}. Simulation results are presented in Section \ref{sec:simulation} and conclusion comes in Section \ref{sec:conclusion}. 

\section{System Model}\label{section:model}
Consider a general single-server  system with $N$ users as shown in Fig. \ref{fig:system}. The server can simultaneously transmit multiple packets to different users with cost due to resource expenditure, e.g.,  energy consumption. The channels are unreliable and transmissions may fail. Each packet has a hard deadline within which it must be delivered successfully. Otherwise, it becomes outdated and will be useless for the user (and discarded). We assume that time is discrete, i.e., $t\in\{0,1,\dots\}$, and a packet transmission to any user takes one time-slot.

\begin{figure}[ht]
\centering
%\vspace{-.1in} 
\includegraphics[width=2.5in]{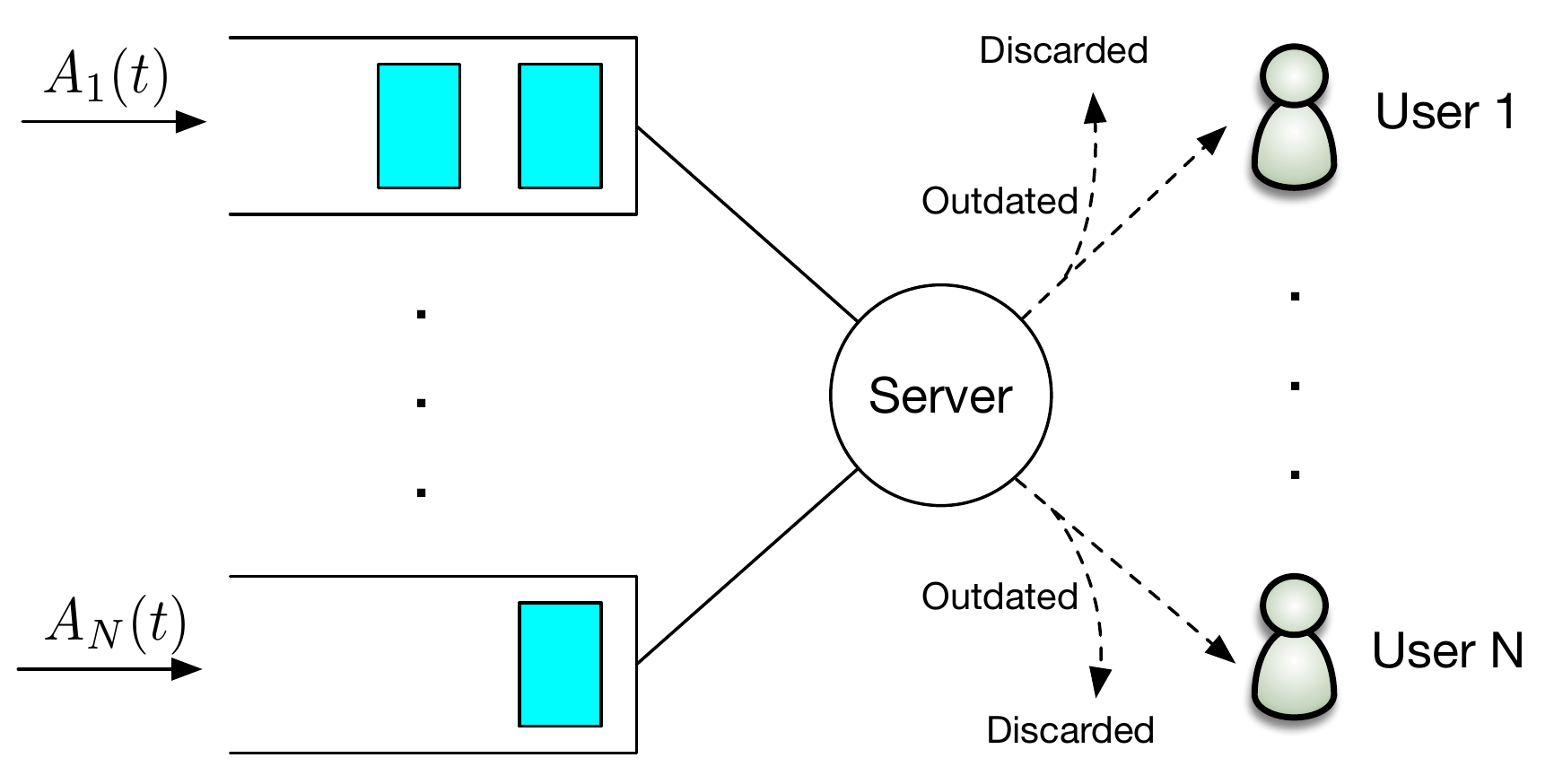}
%\vspace{-.3in}
\caption{A single-server multi-user system where packets have hard deadlines.}
\label{fig:system} 
\vspace{-.1in}
\end{figure}

\subsection{The Delay-Constrained Traffic Model}
The number of packet arrivals destined for user $n$ at time $t$ is denoted by $A_{n}(t)$. We assume that $A_{n}(t)$ is i.i.d across time and independent for different users, with an average rate $\Ex\{A_{n}(t)\}=a_{n}$. We also assume the number of packet arrivals is bounded for all time and for all users, i.e., $0\le A_{n}(t)\le A_{\max}, \forall\, n, t$. 
For each user $n$, there is a hard deadline or sustainable delay for his packets, denoted by $\tau_n$. This means that for any  packet in $A_{n}(t)$, it should be successfully delivered by time $t+\tau_n$.  Otherwise, it becomes useless and will be discarded from the system at time $t+\tau_n$. We further assume $\tau_n\le\Gamma,\forall n$, for some finite constant $\Gamma$. 
% 
%For simplicity of expositio n, let the the packets from user $n$ have the same deadline, defined as $\tau_n$, which  

\subsection{The Service Model} \label{sec:service-model}
The system serves user packets by transmitting them over service channels, at the expense of resource consumption, e.g., energy. 
%As mentioned before, the channels are unreliable. 
To model system dynamics, we assume the success of a packet transmission for user $n$ is a random event and its probability is determined by the instantaneous condition of the service  channel, denoted by the channel state $S_{n}(t)$, which is modeled by an ergodic finite-state Markov chain with state space $\mathcal{S}\triangleq\{s_{1},\dots,s_{K}\},\forall n$. The transition matrix and the stationary distribution are denoted as $(P_{n}^{i,j})_{K\times K}$ and $\bm{\eta}_{n}=(\eta_{n}^{1},\dots,\eta_{n}^{K})$, respectively.   %==== Markov? ====

%state and the consumed resource. The time-varying channel of user $n$ is modeled by an ergodic finite-state Markov chain, whose state is denoted by $S_{n}(t)$. We assume $S_{n}(t)\in\{s_{1},\dots,s_{K}\},\forall n$. For channel states of user $n$, the transition matrix and the stationary distribution are defined as $(P_{n}^{i,j})_{K\times K}$ and $\bm{\eta}_{n}=(\eta_{n}^{1},\dots,\eta_{n}^{K})$, respectively.  

%We denote the resource consumption level for transmitting a user $n$ packet at time $t$ by $E$. 
At every time $t$, the server needs to decide the resource consumption level for transmitting each present packet, which is chosen from a bounded set of consumption levels $\mathcal{E}$.
If at time $t$, the channel state is $s_{i}$ and the resource level is $e$, then the probability of a successful packet transmission for user $n$ is $\zeta_{n}(i,e)$. 
Here $e=0$ means that a packet will not be transmitted in the current time-slot and $\zeta_{n}(i, 0)=0$. Also, $\zeta_n(i, e)>0$  for all $e>0$. 
We further assume that $\zeta_{n}(i,e)$ is a concave and strictly increasing function of $e$. 
% 
%At every time $t$, the consumption level $E$ for each packet is chosen from a bounded set of consumption levels $\mathcal{E}$. %, which will significantly influence the outcome of the transmission.  
%

We assume without loss of generality that  there is a total order on set $\mathcal{S}$ based on $\zeta_{n}(i, e)$, i.e., for each pair of $i, j $, either $\zeta_{n}(i, e)\ge \zeta_{n}(j,e),\, \forall e$ or $\zeta_{n}(i, e)\le \zeta_{n}(j, e), \forall e $. % =I move this here=}
We also assume that there is no hard capacity constraint for the server, i.e., it can transmit an arbitrary number of packets every time, although it has to maintain an average resource consumption guarantee (the setting with hard capacity constraint will be considered in Section \ref{sec:general}). This assumption is made to facilitate analysis and was also adopted in \cite{singh2016throughput}. 

\subsection{The Predictive Service Model}  
Different from  prior results in the literature that often only consider \emph{causal} systems, we are interested in understanding how prediction  and predictive-service fundamentally impact system performance.  
%our model takes account of predictive scheduling, even with prediction error, which enables us to utilize recent advancement in data mining for learning user behavior patterns. 
% 
%To study how prediction and predictive service fundamentally impacts performance, 
Thus, we assume that the system gets access to a \emph{prediction} window $\mathcal{D}_{n}(t) = \{A_{n}(t+1), \dots, A_{n}(t+D_{n})\}$ for each user. 
Moreover, the system implements \emph{predictive service}, i.e., it tries to pre-serve future arrivals in  $\mathcal{D}_{n}(t)$ in the current time slot. 
Such scenario is common in practice. For instance, Amazon  predicts user behavior and pre-ships goods to distribution centers closest to users \cite{amazon-pre-ship-14}.

In this work, we focus on  two prediction models. %, which are commonly adopted in the literature. 
%In this work, we will study the system that can utilize prediction, i.e., the server is allowed to predict and serve future packet arrivals. Specifically, we will consider the following two cases.
\subsubsection{Perfect prediction}
In this case, the predicted arrivals in $\mathcal{D}_{n}(t)$ are exact. 
%the prediction about user $n$  is perfect in window $D_{n}$. That is, at each time $t$, the server can accurately predict the arrival in the lookahead window $\{A_{n}(t+1), \dots, A_{n}(t+D_{n})\}$, and can pre-serve future arrivals in the current time slot.  
This  is an idealized case and results in  this case will serve as an upper bound for the benefit of predictive scheduling. Such a perfect prediction model has been used in the literature, eg., \cite{huang2016backpressure} and \cite{haoran-predictive-16}.  

\subsubsection{Imperfect prediction}  
In the second model,  prediction made by the system can contain error. Specifically, we adopt the imperfect prediction model parameterized by the true-positive and false-negative rates as follows. 
Each predicted arrival in the prediction window is correct with probability $p_{n}$, and every actual packet arrival will be missed with probability $q_{n}$, i.e., a packet will arrive unexpectedly with probability $q_{n}$, as illustrated in Fig. \ref{fig:prediction-model}. Thus, the true-positive rate is $p_{n}$ and the false-negative rate is $q_{n}$. These two rates are decided by the learning methods used to forecast future arrivals and our analysis holds for general $p_{n}$ and $q_{n}$.
Without loss of generality, we  assume $p_{n}>q_{n}$.\footnote{Otherwise, one can inverse $p_{n}$ and $q_{n}$ to ensure that a positive prediction is more likely to become a true arrival than a negative prediction.} %, i.e., the true-positive rate is larger than true 
Note that the perfect prediction case corresponds to $p_n=1$ and $q_n=0$. 
This model was previous adopted in \cite{huang2016backpressure} and \cite{zhang-predict-ton} for studying the delay reduction due to prediction. %, and is consistent with the frame-based prediction approaches often used in practical system operation, e.g., \cite{give some didi example where predict every half an hour?}.  

\begin{figure}[ht]
\centering
\vspace{-.15in}
\includegraphics[width=3.2in]{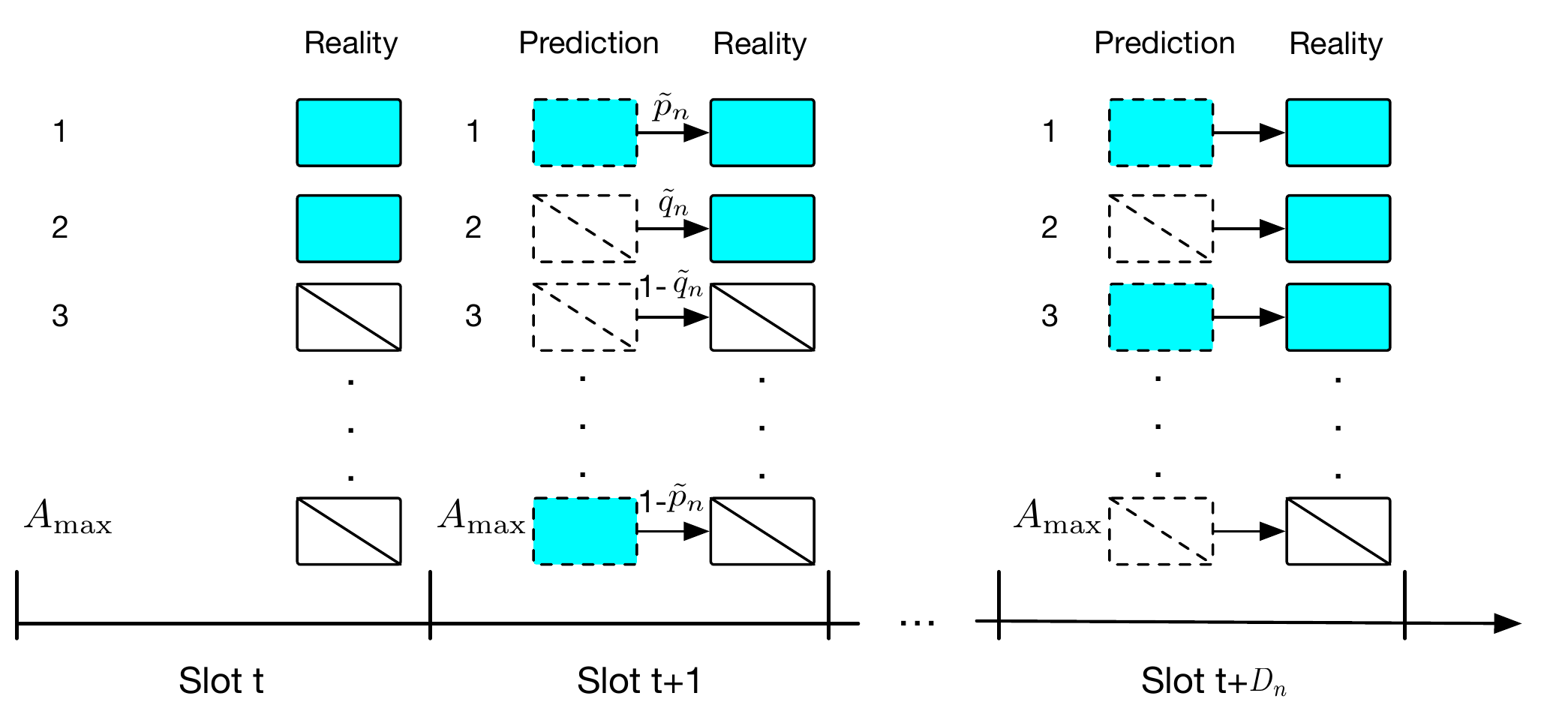}
\vspace{-.1in}
\caption{Imperfect predictions for user $n$ with prediction window $\mathcal{D}_{n}(t)$. The dashed packets are predictions and the solid packets are the actual outcome.  
There are $A_{\max}$ possible arrivals in each time-slot and the server  makes a prediction for each packet. The (dashed) blue box means that the server predicts an arrival (positive prediction) and the correct probability is $p_{n}$. 
The (dashed) crossed white box means that the server predicts no arrival (negative prediction), but a packet may come with probability $q_{n}$. 
At slot $t$,  predicted arrivals will enter the system and the server sees the actual  realizations. 
} 
\label{fig:prediction-model}
%\vspace{-.2in}
\end{figure}

 \vspace{-.15in}
\subsection{System Objective}  
%Since the packets become outdated after the expiry time, w
We define the \emph{timely-throughput} as the average number of packets delivered successfully before their deadlines, i.e.,\footnote{In this paper, we assume all limits in consideration exist with probability $1$. The more general case can be tackled with $\limsup$ or $\liminf$ arguments.} 
\begin{equation}
x_{n}=\lim_{T\to \infty}\frac{1}{T}\expect{\sum_{t=1}^{T}X_{n}(t)},
\end{equation} 
where $X_{n}(t)$ denotes the number of packets  that timely reach their destinations for user $n$ at time $t$. 
We also define the average resource expenditure as: 
\begin{equation}
E_{av}=\lim_{T\to\infty}\frac{1}{T}\expect{\sum_{t=1}^{T}\sum_{n=1}^{N}E_{n}(t)}, 
\end{equation} 
where $E_{n}(t)$ is the resource consumed by transmissions of packets for user $n$ at time $t$. 

Denote $\bm{x}=(x_{1},\dots,x_{N})$. Given a weight vector $\bm{\beta}=(\beta_{1},\dots,\beta_{N})$ with $\beta_{n}\ge0,\forall\, 1\le n\le N$, the \emph{weighted timely-throughput} $\phi$ is defined as $\phi\triangleq\bm{\beta}^{\intercal}\bm{x}$. 
In this paper, we focus on the problem of maximizing the weighted timely-throughput, subject to an average resource constraint $B$, i.e., % is constrained by a budget $B$, i.e., 
\begin{eqnarray}
\phi^*\triangleq\max &&\bm{\beta}^{\intercal}\bm{x}, \label{eq:problem}\\
%\lim_{T\to\infty}\Ex\sum_{n=1}^{N}\sum_{t=1}^{T}\beta_{n}X_{n}(t), \\
\text{s.t.}&&E_{av}\le B. \label{eq:constraint}
\end{eqnarray}
%Here $B$ is the resource budget constraint. 
This formulation is general and models many delay-constrained applications, e.g., video streaming and supply chain optimization.

 \vspace{-.1in}
\subsection{Model Discussion}
%Note that this model has also been adopted in \cite{singh2016throughput} and \cite{I-hong}.  Indeed, 
Our work builds upon the novel results in recent work  \cite{singh2016throughput}. However, our model and results are different as follows. (i) We consider a Markov model for the system while \cite{singh2016throughput} focuses on an i.i.d. setting. 
(ii) We focus on prediction and predictive scheduling and quantify their fundamental impact on timely-throughput, while previous delay-constrained results, e.g., \cite{singh2016throughput} and \cite{hou-delay-10},  mostly consider causal systems. 
Analysis for predictive systems is significantly complicated by potential errors in prediction and requires different arguments compared to those for causal systems. 

Understanding how prediction impacts delay-constrained services is critical for future intelligent communication and computing systems,  as providing delay-guaranteed services has long been an important problem in various applications and predictive scheduling has been successfully utilized in different delay-sensitive scenarios, such as video streaming \cite{zhang2009intelligent} and supply chain optimization \cite{aviv2001effect}.  

% \vspace{-.1in}
\section{Scheduling by Packet-Level Decomposition}\label{section:mdp}
Problem \eqref{eq:problem} can be formulated as a constrained Markov Decision Process (MDP). However, it is known that the number of system states can grow exponentially large, making it complicated to obtain efficient algorithms. 
To tackle this issue, we adopt the packet-level decomposition approach  in \cite{singh2016throughput} and extend it to handle prediction in our setting. 
Specifically, for every individual packet still in the system at time $t$, its state is  described by the user it belongs to and a triple $(r,\tau,i)$. Here $r$ denotes the reception status of the packet, i.e., $r=0$ means that the packet is at the source and $r=1$ means that the packet has reached the destination. 
$\tau$ is the time duration before reaching its  deadline, and $i$ is the channel state index. 
Then, the state of the system at time $t$ can be described by the state of all packets. 
Note that the number of arrivals at any time-slot from any user is bounded by $A_{\max}$,  each packet can stay in the system for no more than $\Gamma$ slots, and the number of channel states is finite, so the number of  system states is finite (though can be exponentially large).  

A (possibly randomized) scheduling policy $\pi$ decides at each system state,  which packets to transmit and at what resource levels. Since the distribution of system state at time $t+1$ is decided by the state and the scheduling decision at time $t$, problem \eqref{eq:problem} is a constrained MDP with finite states, which can be solved with algorithms such as value iteration or policy iteration \cite{bertsekasdptbook}.

%However, although we can merge the states of some packets, the set of system states is still exponentially large, which makes the classical algorithms computationally expensive. Moreover, the classical algorithms hardly give us any insight about the optimal policy or the optimal value of problem \eqref{eq:problem}. So similarly with \cite{singh2016throughput}, we use a packet-level decomposition approach, which can simplify the computation and help us characterize the optimal solution.

\subsection{Packet-Level Decomposition for  the Constrained MDP}
Let $\lambda$ be the Lagrange multiplier for constraint (\ref{eq:constraint}). The Lagrangian of \eqref{eq:problem} can be written as
\begin{eqnarray}
\hspace{-.2in}&&L(\pi,\lambda)= \lim_{T\to\infty}\frac{1}{T}\bigg[\expect{\sum_{n=1}^{N}\sum_{t=1}^{T}\beta_{n}X_{n}(t)} \label{eq:lagrangian} \\
\hspace{-.2in}&&\qquad\qquad\qquad\qquad\qquad-\lambda\expect{\sum_{t=1}^{T}\sum_{n=1}^{N}E_{n}(t)}\bigg]+\lambda B. \nonumber
\end{eqnarray} 
Here $\sum_{n=1}^{N}\sum_{t=1}^{T}\beta_{n}X_{n}(t)$  counts the   timely  deliveries of packets, and $\sum_{t=1}^{T}\sum_{n=1}^{N}E_{n}(t)$ comprises the resource consumed. Further notice that there is no capacity constraint for the server. Thus, by denoting $\mathcal{A}_{n}(T)$ the set of packet arrivals for user $n$ up to time $T$, the Lagrangian \eqref{eq:lagrangian} can be decomposed into the following packet-level form: 
\begin{equation} \label{eq:lagrangian-de}
L(\pi,\lambda)= \lim_{T\to\infty}\frac{1}{T}\expect{\sum_{n=1}^{N}\sum_{\xi\in \mathcal{A}_{n}(T)}[\beta_{n}\delta(\xi)-\lambda E(\xi)]}+\lambda B, 
\end{equation} 
where $\delta(\xi)$ is the indicator that packet $\xi$ reaches the destination before its deadline and $E(\xi)$ is the total resource consumed by packet $\xi$. %, and $A_{n}^{T}$ is the set of packets from user $n$ which arrive before $T$.

From \eqref{eq:lagrangian-de}, the term related to packet $\xi$ of user $n$ is
\begin{equation}\label{eq:lagrangian-oneterm}
\expect{ \beta_{n}\delta(\xi)-\lambda E(\xi)}.
\end{equation} 
As a result, maximizing the Lagrangian \eqref{eq:lagrangian} can be accomplished by maximizing  \eqref{eq:lagrangian-oneterm} for each packet. %, thus yields the Single Packet Scheduling Problem.
In the following, we refer to problem \eqref{eq:lagrangian-oneterm} as the \emph{Single Packet Scheduling Problem} (SPS) and describe how this problem can be solved in the presence of prediction and predictive-service. 

We will carry out our analysis of the SPS problem with a fixed $\lambda$ value in Section \ref{subsec:single-packet}, based on which in Section \ref{subsec:optimal-throughput} we can determine the optimal $\lambda$ and achieve the optimal weighted timely-throughput. 

\subsection{The Single Packet Scheduling Problem} \label{subsec:single-packet}
In this subsection, we consider the optimal solution to the SPS problem under a fixed  multiplier $\lambda$.  
%and consider the optimal scheduling of a single packet. 
%The term \eqref{eq:lagrangian-oneterm} indicates that for any fixed $\lambda$, the optimal scheduling of a packet from user $n$ can be achieved by considering a lightweighted MDP. 
Recall that the state of a packet is described by a triple $(r,\tau,i)$. At each time-slot, the time-to-deadline $\tau$ is decremented by one. If a packet is still at the source when $\tau$ becomes $0$, it will be discarded from the system. On the other hand, if a packet is delivered successfully before the deadline, we collect a reward $\beta_{n}$. The cost charged for resource  expenditure in each transmission is $\lambda$ per unit. 
%  
%This MDP (with prediction) can be solved by Dynamic Programming, no matter the predictions are perfect or imperfect.
 
\subsubsection{Perfect prediction} 
We start with the perfect prediction case (note that zero prediction corresponds to  having $D_n=0$). In this case, the arrival of each packet from user $n$ is known $D_{n}$ timeslots in advance. Thus, we need to solve the SPS problem with an extended deadline of $\tau_n+D_n$. We can define $V_{n}(r,\tau,i)$ as the optimal value function for a packet of user $n$ at state $(r,\tau,i)$. 
The value function and the optimal scheduling decision at each state can be obtained with the following Bellman equations. 
\begin{eqnarray}
\hspace{-.3in}&&V_{n}(0,\tau,i)=\max_{e}\big\{-\lambda e+\zeta_{n}(i,e)\sum_{j}P_{n}^{i,j}V_{n}(1,\tau-1,j)\nonumber\\
\hspace{-.3in}&&\qquad\qquad\qquad+(1-\zeta_{n}(i,e))\sum_{j}P_{n}^{i,j}V_{n}(0,\tau-1,j)\big\},\nonumber\\
\hspace{-.3in}&&\qquad\qquad\qquad0<\tau\le \tau_n+D_{n}, \label{eq:dp} \\
\hspace{-.3in}&&V_{n}(0,0,i)=0,\forall n,i, \label{eq:dp1}\\
\hspace{-.3in}&&V_{n}(1,\tau,i)=\beta_{n},\forall n, i, 0\le\tau<\tau_n+D_n. \label{eq:dp2}
\end{eqnarray}
%Note that the case of no prediction can also be solved with (\ref{eq:dp}) by having $D_n=0$. 

%Now consider packet scheduling with perfect prediction.  We can still use \eqref{eq:dp} to get the solution. The only difference is that the initial time-to-deadline is $\tau_n+D_{n}$ rather than $\tau_n$.

\subsubsection{Imperfect prediction} 
The imperfect prediction case is more complicated. 
We tackle this case by dividing the arrivals into two categories, namely,  the true-positive part and the false-negative part. 
The latter part contains the unpredicted true arrivals. Since these arrivals are not expected, they can only be served after they enter the system.  Thus,  scheduling decisions for these packets are the same as those in the zero prediction case. 

For predicted arrivals, the server can pre-serve them while they are still in the prediction window. % (although the reward even the packet is successfully delivered in the prediction window. ==== 
However, there is a complication in this case. If a predicted arrival is actually a false-alarm,  we cannot collect any reward. As a result, the resource consumed to pre-serve the packet is wasted. 
Moreover, the correctness of a prediction can only be verified at the time when the predicted packet is supposed to arrive. 
Before that, the server will have to take chances and treat all predictions  equally. 

Based on the above reasoning, we will treat predicted packets and the mis-detections differently in the DP formulation. The optimal predictive-service can be done by the following augmented Bellman equations.  
\begin{eqnarray}
\hspace{-.3in}&&V_{n}(0,\tau,i)=\max_{e}\{-\lambda e  \label{eq:dp-aug}\\
\hspace{-.3in}&&\qquad\qquad\qquad +\zeta_{n}(i,e)\sum_{j}P_{n}^{i,j}V_{n}(1,\tau-1,j)\nonumber\\
\hspace{-.3in}&&\qquad\qquad\qquad+(1-\zeta_{n}(i,e))\sum_{j}P_{n}^{i,j}V_{n}(0,\tau-1,j)\},\nonumber\\
\hspace{-.3in}&&\qquad\qquad\qquad0<\tau\le \tau_n+D_{n} ,\tau\ne \tau_n+1, \nonumber\\
\hspace{-.3in}&&V_{n}(0,\tau_n+1,i)=\max_{e}\{-\lambda e \nonumber\\
\hspace{-.3in}&&\qquad\qquad\qquad+\zeta_{n}(i,e)\sum_{j}P_{n}^{i,j}V_{n}(1,\tau_n,j)\label{eq:dp-aug-1}\\
\hspace{-.3in}&&\qquad\qquad\qquad+(1-\zeta_{n}(i,e))p_{n}\sum_{j}P_{n}^{i,j}V_{n}(0,\tau_n,j)\}, \nonumber\\
\hspace{-.3in}&& V_{n}(0,0,i)=0,\forall n,i, \label{eq:dp-aug-2}\\
\hspace{-.3in}&&V_{n}(1,\tau,i)=\beta_{n},\forall n, i, 0\le\tau<\tau_n, \label{eq:dp-aug-3}\\
\hspace{-.3in}&&V_{n}(1,\tau,i)=p_{n}\beta_{n},\forall n,, i, \tau_n\le\tau<\tau_n+D_{n}.\label{eq:dp-aug-4}
\end{eqnarray}
Here  (\ref{eq:dp-aug}) and (\ref{eq:dp-aug-3}) are for unpredicted arrivals, and (\ref{eq:dp-aug-1}) and (\ref{eq:dp-aug-4}) are for predicted packets. 
Compared to  (\ref{eq:dp}), the main difference is that for predicted packets, one needs to take into account the fact that the system will collect a reward $\beta_n$ from pre-serving a packet only with probability $p_n$. 
Also note that the time $\tau_n+1$ is special, as in the next time slot we will be able to verify the correctness of a predicted arrival. Hence, it is separately treated in (\ref{eq:dp-aug-1}). 

%== more explanation ==

\subsection{The Optimal Weighted Timely-Throughput} \label{subsec:optimal-throughput}
After solving the SPS problem for a fixed $\lambda$, the policy $\pi^{*}(\lambda)$ that maximizes the Lagrangian \eqref{eq:lagrangian}  can be derived by letting each packet take its own optimal scheduling decision. 
Next we describe how to optimize the overall problem. 

%% ====stop here  ===

Denote $g(\lambda)$  the Lagrange dual function, i.e., 
\begin{equation}
g(\lambda)=\max_{\pi}L(\pi,\lambda).
\end{equation} 
Using Lemma 3 in \cite{singh2016throughput}, the optimal weighted timely-throughput $\phi^{*}$ equals the optimal value of the  dual problem, i.e., 
\begin{equation} 
\phi^{*}=\min_{\lambda\ge0}g(\lambda). \label{eq:zero-duality-gap}
\end{equation} 
This can be established by showing that the constrained MDP, with or without prediction, is equivalent to a linear program \cite{altman1999constrained},\cite{altman1996constrained}. Hence, the duality gap is zero. 

%where the optimal Lagrange multiplier $\lambda^{*}$ can be obtained by subgradient descent. 
%The optimal weighted timely-throughput problem can then be solved by finding a policy to maximize $g(\lambda)$. == this has been proven ?? ==

In the following, we look at the Lagrange dual function in the perfect and imperfect prediction cases. 
%different cases.

\subsubsection{Dual under zero or perfect prediction}  
Let $\bm{V}_{n}\triangleq (V_{n}(0,\tau_n+D_{n},1), V_{n}(0,\tau_n+D_{n},2),\dots,V_{n}(0,\tau_n+D_{n},K))$ in the case with perfect prediction (zero prediction corresponds to $D_n=0$). Then, using (\ref{eq:lagrangian-de}), the dual function can be expressed as (recall that $\bm{\eta}_{n}$ is the steady-state distribution of the channel state for user $n$): 
\begin{equation} \label{eq:dual-noprediction}
g(\lambda)=\sum_{n=1}^{N}a_{n}\bm{\eta}_{n}^{\intercal}\bm{V}_{n}+\lambda B. 
\end{equation} 

\subsubsection{Dual under imperfect prediction} 
In this case, we note that the rate of predicted arrivals may not equal the actual arrival rate. 
Denote the predicted arrival rate from user $n$ as $\tilde{a}_{n}$. We have: 
\begin{equation*}
%a_{n}=\tilde{a}_{n}p+(a_{\max}-\tilde{a}_{n})q.
a_{n}=\tilde{a}_{n}p_{n}+(a_{\max}-\tilde{a}_{n})q_{n}.
\end{equation*}
Here the second term is due to the fact that each packet is missed, i.e., not predicted, with probability $q_n$. Thus,  %==need more thoughts==
\begin{equation}\label{eq:pre-arrival-rate}
%\tilde{a}_{n}=\frac{a_{n}-a_{\max}q}{p-q}.
\tilde{a}_{n}=\frac{a_{n}-a_{\max}q_{n}}{p_{n}-q_{n}}.
\end{equation}
Since $p_{n}>q_{n}$ and $0\le\tilde{a}_{n}\le a_{\max}$, we  get
\begin{equation}
q_{n}\le\frac{a_{n}}{a_{\max}}\le p_{n}.
\end{equation}
Let $\bm{V}_{n}\triangleq(V_{n}(0,\tau_n,1), V_{n}(0,\tau_n,2),\dots,V_{n}(0,\tau_n,K))$ and $\tilde{\bm{V}}_{n}\triangleq(V_{n}(0,\tau_n+D_{n},1), V_{n}(0,\tau_n+D_{n},2), \dots, V_{n}(0,\tau_n+D_{n},K))$. Similar to the perfect prediction case, the dual function can be expressed as:  
\begin{equation} \label{eq:dual-imperfectprediction}
g(\lambda)=\sum_{n=1}^{N}\left[\tilde{a}_{n}\bm{\eta}_{n}^{\intercal}\tilde{\bm{V}}_{n}+(a_{\max}-\tilde{a}_{n})q_{n}\bm{\eta}_{n}^{\intercal}\bm{V}_{n}\right]+\lambda B, 
\end{equation}
where $\tilde{a}_{n}$ is determined in \eqref{eq:pre-arrival-rate}. % and the 

After obtaining the dual function for a fixed $\lambda$, we still need to find the optimal Lagrange multiplier $\lambda^{*}=\argmin_{\lambda\ge0}g(\lambda)$. One possible approach is to use the subgradient descent method, where we take an iterative procedure to converge to $\lambda^{*}$ as follows. 
In the $k$-th iteration, we solve the SPS problem to get the optimal policy $\pi^{*}(\lambda_{k})$ and the average resource expenditure $E_{av}(\pi^{*}(\lambda_{k}))$ based on the current multiplier $\lambda_{k}$. 
Then, the multiplier for the next iteration is given by ($\epsilon_k$ is a step size): 
\begin{equation*}
\lambda_{k+1}=\lambda_{k}+\epsilon_k[E_{av}(\pi^{*}(\lambda_{k}))-B],
\end{equation*}
It is known that with an appropriately chosen $\{\epsilon_k\}_{k=1}^{\infty}$ sequence, $\lambda_k\rightarrow\lambda^*$ \cite{bertsekas2003convex}. 

Despite the generality of (\ref{eq:dual-noprediction}) and (\ref{eq:dual-imperfectprediction}), directly solving them is complicated. Thus, in the following section, we first consider a slightly less general setting where user channels are static (can be different across users) and the resource expenditure option is binary.\footnote{Due to the complexity of the general stochastic scenario, similar static settings were also considered in the literature, e.g., \cite{rechargeable-ton-14} and \cite{neely-finite-queue-17}.} 
Results for the general case will be presented in Section \ref{sec:general}.

\section{The Static Scenario} \label{sec:simple-case} 
In this scenario, we  assume that the channel states are static, i.e., the success probability for transmitting a user $n$ packet is a constant $\zeta_{n}$. 
Moreover, we assume the resource level set is $\mathcal{E}=\{0, 1\}$, i.e., at each time-slot,  the scheduling decision for each packet is to transmit it  or not. 
In this case, the state of each packet can be described by $(r,\tau)$.   
% 

%\subsection{Optimal Policy for a Fixed $\lambda$}
\subsection{The Optimal Scheduling Policy}
%For a fixed $\lambda$, we can  explicitly obtain the optimal policies in different cases.
\subsubsection{Perfect prediction} 
First we consider the perfect prediction case. 
We have the following theorem. Recall that the zero prediction case is a special case ($D_n=0$ for all $n$). 
\begin{theorem}\label{thm:policy-simple-nopre}
For each user $n$ packet, if $\zeta_{n}\beta_{n}>\lambda$, then the optimal policy is to transmit the packet at every time-slot, until it is either successfully delivered to the destination or becomes outdated. Moreover, the value function is given by: 
\begin{equation}
V_{n}(0,\tau)=\frac{1-(1-\zeta_{n})^{\tau}}{\zeta_{n}}(-\lambda+\zeta_{n}\beta_{n}), 0<\tau\le \tau_n+D_n. \label{eq:perfect-predict-value-fun}
\end{equation} 
Otherwise, if $\zeta_{n}\beta_{n}\le\lambda$,  the value function is $V_{n}(0,\tau)=0, 0<\tau\le \tau_n+D_n$. Specially, if $\zeta_{n}\beta_{n}<\lambda$, the optimal policy is to not transmit the packet at all.
\end{theorem}
\begin{proof}
See Appendix \ref{apd:policy-simple-nopre}.
\end{proof}
\begin{remark} 
When $\lambda=\lambda^{*}$, based on the KKT conditions \cite{bertsekas2003convex}, it can be shown that  for packets with $\zeta_{n}\beta_{n}=\lambda^{*}$,  the optimal policy is to transmit the packet at every time-slot with probability $\bar{p}_{P}=\frac{B-\sum_{j:\zeta_{j}\beta_{j}>\lambda^{*}}a_{j}\bar{E}_{j}}{a_{n}\bar{E}_{n}}$, where $\bar{E}_{n}=\zeta_{n}+2(1-\zeta_{n})\zeta_{n}+\cdots+(\tau_n+D_{n})(1-\zeta_{n})^{\tau_n+D_{n}-1}$, and not to transmit the packet  otherwise. Notice that this has no influence on the value of the dual function $g(\lambda)$, as well as the optimal weighted timely-throughput obtained by $\phi^{*}=g(\lambda^{*})$. 
\end{remark}
Theorem \ref{thm:policy-simple-nopre} shows that if the expected reward is larger than the cost in one transmission, then the server should try its best to deliver packets for user $n$. Otherwise, packets from user $n$ should never be served if the expected reward is less than the cost (from the point of optimizing the weighted throughput under an average resource constraint). 
From Theorem \ref{thm:policy-simple-nopre} and (\ref{eq:zero-duality-gap}), we have the following corollary. 
\begin{coro}\label{coro:throughput-simple-nopre}
Let $g_{P}(\lambda)$ denote the the dual function of \eqref{eq:problem} and $\phi_{P}^{*}$ denote the optimal weighted timely-throughput in the perfect  prediction case. We have: 
\begin{equation}
\begin{aligned}
g_{P}(\lambda)=&\max_{\pi}L(\pi,\lambda)\\
=&\sum_{n:\zeta_{n}\beta_{n}>\lambda}a_{n}\frac{1-(1-\zeta_{n})^{\tau_n+D_n}}{\zeta_{n}}(-\lambda+\zeta_{n}\beta_{n})+\lambda B,
\end{aligned}
\end{equation}
and 
\begin{equation} \label{eq:throughput-simple-nopre}
\phi^{*}_{P}=\min_{\lambda}\sum_{n:\zeta_{n}\beta_{n}>\lambda}a_{n}\frac{1-(1-\zeta_{n})^{\tau_n+D_n}}{\zeta_{n}}(-\lambda+\zeta_{n}\beta_{n})+\lambda B.
\end{equation}
\end{coro}

%Since the perfect prediction case is equivalent to extending the deadlines of packets,  we have the following theorem. %  similar conclusions.
%\begin{theorem}\label{thm:policy-simple-perpre} 
%Suppose the prediction window for user $n$ is $D_{n}>0$ and the predictions are perfect. For each packet from user $n$, if $p_{n}\beta_{n}>\lambda$, then the optimal policy is to transmit the packet at each time-slot as soon as it enters the prediction window, until it is successfully delivered to the destination or outdated, and the value function is given by
%\begin{equation}
%V_{n}(0,\tau)=\frac{1-(1-p_{n})^{\tau}}{p_{n}}(-\lambda+p_{n}\beta_{n}), 0<\tau\le \tau_n+D_{n}.
%\end{equation}
%On the other hand, if $p_{n}\beta_{n}\le\lambda$, then the optimal policy is to not transmit the packet at all, and the value function is $V_{n}(0,\tau)=0, 0<\tau\le \tau_n+D_{n}$.
%\end{theorem}
%\begin{coro}\label{coro:throughput-simple-perpre}
%Let $g_{P}(\lambda)$ denote the the Lagrange dual function of \eqref{eq:problem} and $x_{P}^{*}$ denote the optimal weighted timely-throughput in the perfect prediction case. We can get
%\begin{equation}
%\begin{aligned}
%g_{P}(\lambda)=&\max_{\pi}L(\pi,\lambda)\\
%=&\sum_{n:p_{n}\beta_{n}>\lambda}a_{n}\frac{1-(1-p_{n})^{\tau_n+D_{n}}}{p_{n}}(-\lambda+p_{n}\beta_{n})+\lambda B,
%\end{aligned}
%\end{equation}
%and 
%\begin{equation} \label{eq:throughput-simple-perpre}
%x^{*}_{P}=\min_{\lambda}\sum_{n:p_{n}\beta_{n}>\lambda}a_{n}\frac{1-(1-p_{n})^{\tau_n+D_{n}}}{p_{n}}(-\lambda+p_{n}\beta_{n})+\lambda B.
%\end{equation}
%\end{coro}

Corollary \ref{coro:throughput-simple-nopre}  enables us to characterize the fundamental improvement in weighted timely-throughput due to prediction, which is shown in the next theorem. In the theorem,  $g_{0}(\lambda)$  and $\phi^*_0$ denote the dual function and optimal weighted timely-throughput without prediction, respectively. 
\begin{theorem}\label{thm:improve-simple-perpre}
Suppose $\lambda_{0}^{*}=\argmin_{\lambda} g_{0}(\lambda)$ and $\lambda_{P}^{*}=\argmin_{\lambda} g_{P}(\lambda)$, then the weighted timely-throughput improvement satisfies:  
\begin{eqnarray}
\hspace{-.5in}&&\phi^{*}_{P}-\phi^{*}_{0}\le\sum_{n:\zeta_{n}\beta_{n}>\lambda_{0}^{*}}\frac{a_{n}}{\zeta_{n}}\left[(1-\zeta_{n})^{\tau_n}-(1-\zeta_{n})^{\tau_n+{D_{n}}}\right]\nonumber\\
\hspace{-.5in}&&\qquad\qquad\qquad\qquad\qquad\times(-\lambda_{0}^{*}+\zeta_{n}\beta_{n})
\end{eqnarray}
and
\begin{eqnarray}
\hspace{-.5in}&&\phi^{*}_{P}-\phi^{*}_{0}\ge\sum_{n:\zeta_{n}\beta_{n}>\lambda_{P}^{*}}\frac{a_{n}}{\zeta_{n}}\left[(1-\zeta_{n})^{\tau_n}-(1-\zeta_{n})^{\tau_n+{D_{n}}}\right]\nonumber\\
\hspace{-.5in}&&\qquad\qquad\qquad\qquad\qquad\times(-\lambda_{P}^{*}+\zeta_{n}\beta_{n}).
\end{eqnarray}
\end{theorem}
\begin{proof}
From \eqref{eq:throughput-simple-nopre}, we have that:  
\begin{align*}
\phi^{*}_{P}=&\min_{\lambda}g_{P}(\lambda)\le g_{P}(\lambda^{*}_{0}), \\
\phi^{*}_{0}=&\min_{\lambda}g_{0}(\lambda)\le g_{0}(\lambda^{*}_{P}).
\end{align*}
Thus, $g_{P}(\lambda^{*}_{P})-g_{0}(\lambda^{*}_{P})\le \phi^{*}_{P}-\phi^{*}_{0}\le g_{P}(\lambda^{*}_{0})-g_{0}(\lambda^{*}_{0})$, which gives us Theorem \ref{thm:improve-simple-perpre}. 
\end{proof}
%\begin{remark} 
%Notice that $\lambda_{P}^{*}\ge\lambda_{0}^{*}$. This can be seen by comparing the average resource expenditure in two cases. 
%% 
%In the zero prediction case, $E_{av}=\sum_{n:\zeta_{n}\beta_{n}>\lambda}a_{n}[\zeta_{n}+2(1-\zeta_{n})\zeta_{n}+\cdots+\tau_n(1-\zeta_{n})^{\tau_n-1}]$. In the perfect prediction case, $E_{av}=\sum_{n:\zeta_{n}\beta_{n}>\lambda}a_{n}[\zeta_{n}+2(1-\zeta_{n})\zeta_{n}+\cdots+(\tau_n+D_{n})(1-\zeta_{n})^{\tau_n+D_{n}-1}]$, which is larger than that in the first case if $\lambda$ remains the same. 
%% 
%Thus, we need to increase $\lambda$ in the perfect prediction case to meet the constraint. %\textcolor{red}{ should we put $\lambda_{P}^{*}\ge\lambda_{0}^{*}$ into the proof? }
%%This is because after we introduce prediction into the system, if the scheduling decision is to not preserve the packets, then the average resource consumed per time-slot does not change. Otherwise the average resource consumed per time-slot will increase.   %==should try to formally establish this== 
%%Therefore, $\lambda_{P}^{*}\ge\lambda_{0}^{*}$.
%\end{remark}

Theorem \ref{thm:improve-simple-perpre} suggests two things. (i) for each user $n$, prediction improves the throughput by an amount $\Theta(\rho^{\tau_n}(1-\rho^{D_{n}}))$ with $\rho=1-\zeta_{n}$. 
This implies that the impact of prediction is decreasing with the deadline $\tau_n$, which is expected. (ii) the gap to the optimal improvement $\rho^{\tau_n}$ decreases \emph{exponentially} as the prediction power $D_n$ increases. This demonstrates the power of prediction and highlights the potential of investing in improving system prediction. 
%and increases exponentially to $\rho^{\tau_n}$ as the prediction power $D_n$ increases. 

%means that the influence of prediction decreases exponentially with the initial time-to-deadline and polynomially with the successful transmission probability.

\subsubsection{Imperfect prediction}
In this case, we first note that there is not much the system can do with the unpredicted arrivals. 
Thus, the  optimal scheduling policy and value function for these packets are the same as those in Theorem \ref{thm:policy-simple-nopre}. 
%In the imperfect prediction case, if the $\lambda$ is fixed, for those false negative predictions, i.e., unexpected real arrivals, we can do nothing but wait to serve them until them enter the system, and the optimal scheduling policy and value function are given by Theorem \ref{thm:policy-simple-nopre}.  
For  the predicted arrivals,  on the other hand,  the system will be able to start their services  the moment they appear in the prediction window, following (\ref{eq:dp-aug})  to (\ref{eq:dp-aug-4}).  
%the optimal scheduling policy and value function after they enter the system (if the predictions are true) are the same as Theorem \ref{thm:policy-simple-nopre}. 
%Now we consider the optimal preserving policy and the corresponding value function. 

The following theorem characterizes the optimal predictive-service policy and the corresponding value functions. Recall that $p_{n}$ is the true-positive probability  and $q_{n}$ is the false-negative probability. 
\begin{theorem}\label{thm:policy-simple-impre} 
%Suppose the prediction window for user $n$ is $D_{n}>0$, the true positive rate is $p_{n}$, and the false negative rate is $q_{n}$. 
Consider a predicted arrival for user  $n$.

(A) Suppose $\zeta_{n}\beta_{n}>\lambda$. Define  
\begin{eqnarray}
c_{n}\triangleq\frac{\lambda}{(-\lambda+\zeta_{n}\beta_{n})(1-\zeta_{n})^{\tau_n}+\lambda}. 
\end{eqnarray}
%$c_{n}=\frac{\lambda}{(-\lambda+\zeta_{n}\beta_{n})(1-\zeta_{n})^{\tau_n}+\lambda}$. 
(i) If $p_{n}> c_{n}$, then the optimal pre-service policy is to transmit the packet at every time-slot once it enters the prediction window, until it is either successfully delivered to the destination or revealed to be a false-alarm. 
The value function is given by: 
\begin{eqnarray} 
\hspace{-.3in}&&V_{n}(0,\tau_n+w)=-(1-p_{n})\frac{1-(1-\zeta_{n})^{w}}{\zeta_{n}}\lambda\label{eq:value-fun-pq}\\
\hspace{-.3in}&&\qquad\qquad\qquad\qquad  +p_{n}\frac{1-(1-\zeta_{n})^{\tau_n+w}}{\zeta_{n}}(-\lambda+\zeta_{n}\beta_{n}), \nonumber
%\end{aligned}
\end{eqnarray}
where $0<w\le D_{n}$. (ii) If $p_{n}\le c_{n}$, the value function is given by: 
\begin{equation}\label{eq:value-fun-pq2}
\begin{aligned}
V_{n}(0,\tau_n+w)=p_{n}\frac{1-(1-\zeta_{n})^{\tau_n}}{\zeta_{n}}(-\lambda+\zeta_{n}\beta_{n}),
\end{aligned}
\end{equation}
where $0<w\le D_{n}$. Specially, if $p_{n}< c_{n}$, the optimal policy is to not pre-serve the packet and to wait until it enters the system (if it is a true-positive), 

(B) If $\zeta_{n}\beta_{n}\le\lambda$, $V_{n}(0,\tau_n+w)=0, 0<w\le D_{n}$.  The optimal policy is to not transmit the packet at all if $\zeta_{n}\beta_{n}<\lambda$.
\end{theorem}
\begin{proof}
See Appendix \ref{apd:policy-simple-impre}.
\end{proof}

\begin{remark}
Similar with Theorem \ref{thm:policy-simple-nopre}, when $\lambda=\lambda^{*}$, if there exist $n_{1}, n_{2}$ such that $p_{n_{1}}= c_{n_{1}},\zeta_{n_{2}}\beta_{n_{2}}=\lambda^{*}$, then the optimal policy is to preserve $n_{1}$ packets with certain probability $\bar{p}_{I,1}$ and to transmit $n_{2}$ packets with certain probability $\bar{p}_{I,2}$, such that $E_{av}=B$.

Theorem \ref{thm:policy-simple-impre} shows that for a  predicted user $n$ arrival  with $\zeta_{n}\beta_{n}>\lambda$, if the server waits until it enters the system, then does its best to deliver it,  the expected reward is $p_{n}\frac{1-(1-\zeta_{n})^{\tau_n}}{\zeta_{n}}(-\lambda+\zeta_{n}\beta_{n})$. 
This is consistent with our intuition compared with \eqref{eq:perfect-predict-value-fun}, as that the probability that a  predicted arrival is real is $p_{n}$. Also note that although $q_{n}$ does not appear in $V_n(0, \tau)$, we will see in the following corollary that it  indirectly impacts the final timely-throughput by affecting the effective arrival rate as in (\ref{eq:pre-arrival-rate}). 
\end{remark}

Note that $c_{n}$ can intuitively be viewed as the weight put on resource consumption compared to reward. 
Hence, when $p_{n}> c_{n}$, the true-positive rate is large enough such that pre-transmitting the packet in one time slot, i.e., at time-slot $\tau_n+1$, will increase the value function. Under this circumstance, Theorem \ref{thm:policy-simple-impre} shows that the optimal scheduling is to transmit the packet as early as possible. 

%==explain the meaning of $c_n$==  \textcolor{red}{== to be done ==}

For a user $n$ that satisfies $\zeta_{n}\beta_{n}>\lambda$, define
\begin{eqnarray}
v_{n}(\lambda)=
\begin{cases}
&p_{n}\frac{1-(1-\zeta_{n})^{\tau_n}}{\zeta_{n}}(-\lambda+\zeta_{n}\beta_{n}), \,\,p_{n}\le c_{n},\\
&-(1-p_{n})\frac{1-(1-\zeta_{n})^{D_{n}}}{\zeta_{n}}\lambda\\
&+p_{n}\frac{1-(1-\zeta_{n})^{\tau_n+D_{n}}}{\zeta_{n}}(-\lambda+\zeta_{n}\beta_{n}), \,\,p_{n}> c_{n}. 
\end{cases}
\end{eqnarray}
We have the following immediate corollary from Theorem \ref{thm:policy-simple-impre}. 
\begin{coro}\label{coro:throughput-simple-impre}
Let $g_{I}(\lambda)$ and $\phi_{I}^{*}$ denote the the dual function of \eqref{eq:problem} and  the optimal weighted timely-throughput in the imperfect prediction case. We have: 
\begin{eqnarray}
%\begin{aligned}
&&g_{I}(\lambda)=\lambda B+\sum_{n:\zeta_{n}\beta_{n}>\lambda}\big[\tilde{a}_{n}v_{n}(\lambda)\label{eq:gI-imperfect}\\ 
&&\qquad\qquad  +(a_{\max}-\tilde{a}_{n})q_{n}\frac{1-(1-\zeta_{n})^{\tau_n}}{\zeta_{n}}(-\lambda+\zeta_{n}\beta_{n})\big], \nonumber
%\end{aligned}
\end{eqnarray}
and 
\begin{equation} \label{eq:throughput-simple-impre}
\phi^{*}_{I}=\min_{\lambda}g_{I}(\lambda),
\end{equation}
where $\tilde{a}_{n}$ is the rate of predicted arrivals given in \eqref{eq:pre-arrival-rate}.
\end{coro}
Combining  Corollaries \ref{coro:throughput-simple-nopre} and  \ref{coro:throughput-simple-impre}, we conclude the following theorem regarding the improvement in timely-throughput with imperfect prediction. 
\begin{theorem}\label{thm:improve-simple-impre}
Suppose $\lambda_{I}^{*}=\argmin g_{I}(\lambda)$. The weighted timely-throughput improvement satisfies 
\begin{eqnarray}
%\begin{aligned}
\hspace{-.4in}&&\phi^{*}_{I}-\phi^{*}_{0}\le\sum_{n:\zeta_{n}\beta_{n}>\lambda_{0}^{*}}\bigg\{\tilde{a}_{n}v_{n}(\lambda_{0}^{*})\\
\hspace{-.4in}&&\quad  -\frac{(a_{n}p_{n}-a_{\max}p_{n}q_{n})[1-(1-\zeta_{n})^{\tau_n}]}{(p_{n}-q_{n})\zeta_{n}}(-\lambda_{0}^{*}+\zeta_{n}\beta_{n})\bigg\},\nonumber
%\end{aligned}
\end{eqnarray}
and
\begin{eqnarray}
%\begin{aligned}
\hspace{-.4in}&&\phi^{*}_{I}-\phi^{*}_{0}\ge\sum_{n:\zeta_{n}\beta_{n}>\lambda_{I}^{*}}\bigg\{\tilde{a}_{n}v_{n}(\lambda_{I}^{*})\\
\hspace{-.4in}&&\quad -\frac{(a_{n}p_{n}-a_{\max}p_{n}q_{n})[1-(1-\zeta_{n})^{\tau_n}]}{(p_{n}-q_{n})\zeta_{n}}(-\lambda_{I}^{*}+\zeta_{n}\beta_{n})\bigg\},\nonumber
%\end{aligned}
\end{eqnarray}
where $\tilde{a}_{n}$ is given in \eqref{eq:pre-arrival-rate}.
\end{theorem}
\begin{proof}
Similar to the proof of Theorem \ref{thm:improve-simple-perpre}.
\end{proof}
\begin{remark}
Similar to Theorem \ref{thm:improve-simple-perpre}, we still have in the imperfect prediction case that , the improvement scales in the order of $\Theta(p_{n}\left[C_{1}\frac{(a_{n}-a_{\max}q_{n})}{p_{n}-q_{n}}\rho^{\tau_n}+C_{2}(1-\frac{1}{p_{n}})\right](1-\rho^{D_{n}}))$, which recovers the perfect prediction result when $p_{n}=1$ and $q_{n}=0$. This general result shows how different parameters affect the optimal timely-throughput, and will be useful for guiding prediction and control algorithm design for general computing systems. 
\end{remark}

\subsection{Influence of Prediction Accuracy} 
We now investigate how prediction accuracy impacts performance improvement. The following theorem summarizes our results.

\begin{theorem}\label{throughput-simple-p-q}
Let $\bm{q}=\{q_{1},q_{2},\dots,q_{N}\}$ denote the false-negative rate vector, and let $\bm{p}=\{p_{1},p_{2},\dots,p_{N}\}$ be the true-positive rate vector. Then, 

(i) The optimal weighted timely-throughput $\phi^*_I$ is a non-increasing function of $\bm{q}$. 

(ii) If $\bm{q}=\bm{0}$, then the optimal weighted timely-throughput $\phi^*_I$ is a non-decreasing function of $\bm{p}$. 
\end{theorem}
\begin{proof}
See Appendix C. 
\end{proof}
The results, though intuitive, are non-trivial to established and require detailed investigation of the structure of the value functions. The impact of general $(\bm{p}, \bm{q})$ pairs, on the other hand, is much more complicated to characterize, as we will see in Fig. \ref{fig:throughput-surface} in the simulation section.

%\begin{theorem}\label{throughput-simple-p}
%Define the true-positive rate vector $\tilde{\bm{p}}=\{p_{1},p_{2},\dots,p_{N}\}$ and the false-negative rate vector as in Theorem \ref{throughput-simple-q}. If $\tilde{\bm{q}}=\bm{0}$, then the optimal weighted timely-throughput $\phi^*_I$ is a non-decreasing function of $\tilde{\bm{q}}$. 
%\end{theorem}
%\begin{proof}
%See Appendix D. 
%\end{proof}

\section{The General Scenario}\label{sec:general} 
We now return to the general case. We first show that the optimal policy in the perfect prediction case (including zero prediction) is monotone with respect to  the time-to-deadline. 
\begin{theorem} \label{theorem:monotone}
Suppose $\zeta_{n}(i,e)$ is a concave strictly increasing function of $e$. 
In the perfect prediction case, define $e_{n}^{*}(0,\tau,i)$ as the optimal scheduling decision for a user $n$ packet at state $(0,\tau,i)$. We have: 
\begin{equation}\label{eq:policy-monotonicity-general}
e_{n}^{*}(0,\tau+1,i)\le e_{n}^{*}(0,\tau,i), \forall n,i,
\end{equation}
where $0<\tau<\tau_n+D_n$ (where $D_n=0$ denotes the zero prediction case). % in the zero prediction case, and $\tau<\tau_n+D_{n}$ in the perfect prediction case.
\end{theorem}
\begin{proof}
See Appendix \ref{apd:monotone}. 
\end{proof}
Theorem \ref{theorem:monotone} shows that the optimal scheduling policy is a ``lazy'' policy, i.e., the server will not try hard to transmit the packet unless it is getting close to the deadline. From the proof of Theorem \ref{theorem:monotone}, we see that this is because the value function under the same channel state is monotonically non-decreasing with the time-to-deadline. Thus, the server will try to spend less resource at the beginning to see if the packet can luckily get through, and spend more resource when the deadline is getting close, so as to pursue the reward of successful delivery.
This result nicely complements existing efficient scheduling results in the literature \cite{prabhakar2001energy}, \cite{modiano2009minimum}, \cite{modiano2009calculus}, and can be viewed as an extension to the predictive online scheduling setting. 

It is tempting to conclude that similar property also holds in the general imperfect prediction case. However,  \emph{this monotonicity actually does not hold under imperfect prediction.} 
This is because the value function is no longer monotone due to prediction error. 
%\textcolor{red}{=more explanation is needed=} 
This fact will be  illustrated by simulation in Section \ref{sec:simulation}, where we see that the resource level can actual decrease with a smaller time-to-deadline (See User $2$'s strategy in Table \ref{tab:imprediction}).

\subsection{Throughput Improvement}
%For a fixed $\lambda$, we can  explicitly obtain the optimal policies in different cases.
In this subsection, we investigate the throughput improvement due to  prediction in the general scenario.
\subsubsection{Perfect prediction} 
First we consider the perfect prediction case. Define $i_{n}^{\max}\triangleq \argmax_{i}\zeta_{n}(i,e),\forall e$ and $i_{n}^{\min}\triangleq\argmin_{i}\zeta_{n}(i,e),\forall e$. 
Both $i_{n}^{\max}$ and $i_{n}^{\min}$ are well defined since there is a complete order on $\mathcal{S}$ based on $\zeta_{n}(i, e)$ (see Section \ref{sec:service-model}). 
We then have the following theorem (the zero prediction case is a special case, i.e., $D_n=0$ for all $n$). 
\begin{theorem}\label{thm:value-general-perpre}
%For each user $n$ packet, if $\lambda$ is fixed, the value function satisfies: 
%\begin{equation} \label{eq:value-general-perpre-upper}
%V_{n}(0,\tau,i)\le\min(\beta_{n},\tau\max_{E}\{-\lambda E+p_{n}(i_{n}^{\max},E)\beta_{n}\})
%\end{equation} 
%and
%\begin{equation} \label{eq:value-general-perpre-lower}
%\begin{aligned}
%V_{n}(0,\tau,i)\ge\max_{E}&\frac{1-[1-p_{n}(i_{n}^{\min},E)]^{\tau}}{p_{n}(i_{n}^{\min},E)}[-\lambda E\\
%&+p_{n}(i_{n}^{\min},E)\beta_{n}],
%\end{aligned}
%\end{equation}
%where $0<\tau\le \tau_n+D_n$.
For $1\le n\le N,0<\tau\le \tau_n+D_{n}$, define 
\begin{equation} \label{eq:value-general-perpre-lower}
\begin{aligned}
V_{n}^{l}(\tau)\triangleq\max_{e>0}\frac{1-[1-\zeta_{n}(i_{n}^{\min},e)]^{\tau}}{\zeta_{n}(i_{n}^{\min},e)}[-\lambda e+\zeta_{n}(i_{n}^{\min},e)\beta_{n}],
\end{aligned}
\end{equation}
where $V_{n}^{l}(0)=0$, and
\begin{eqnarray} 
\hspace{-.2in}&&V_{n}^{u}(\tau)\triangleq\sum_{z=1}^{\tau}\max_{e}\{-\lambda e\label{eq:value-general-perpre-upper}\\
\hspace{-.2in}&& \qquad\qquad+\zeta_{n}(i_{n}^{\max},e)(\beta_{n}-\max\{0,V_{n}^{l}(z-1)\})\}, \nonumber 
%\end{aligned}
\end{eqnarray}
where $V_{n}^{u}(0)=0$. 
For each user $n$ packet, given a fixed $\lambda$, the value function satisfies: 
\begin{equation} \label{eq:value-general-perpre}
\max\{0,V_{n}^{l}(\tau)\}\le V_{n}(0,\tau,i)\le\min\{\beta_{n},V_{n}^{u}(\tau)\}
\end{equation} 
where $0<\tau\le \tau_n+D_n$.
\end{theorem}
\begin{proof}
See Appendix \ref{apd:value-perpre}.
\end{proof}
%We will 
We then immediately have  the following corollary, which shows that the functions $V_{n}^{l}$ and $V_{n}^{u}$ enable us to bound the optimal timely-throughput with prediction. 
\begin{coro}\label{coro:throughput-general-perpre}
Define
\begin{equation}
\begin{aligned}
%g_{P}^{u}(\lambda)\triangleq&\lambda B+\sum_{n=1}^{N}a_{n}\min(\beta_{n},\\
%&(\tau_n+D_{n})\max_{E}\{-\lambda E+\zeta_{n}(i_{n}^{\max},E)\beta_{n}\})
g_{P}^{u}(\lambda)\triangleq&\lambda B+\sum_{n=1}^{N}a_{n}\min\{\beta_{n},V_{n}^{u}(\tau_n+D_{n})\}
\end{aligned}
\end{equation}
and 
\begin{equation}
\begin{aligned}
%g_{P}^{l}(\lambda)\triangleq&\lambda B+\sum_{n=1}^{N}a_{n}\max_{E}\frac{1-[1-\zeta_{n}(i_{n}^{\min},E)]^{\tau_n+D_{n}}}{\zeta_{n}(i_{n}^{\min},E)} \\
%&\times[-\lambda E+\zeta_{n}(i_{n}^{\min},E)\beta_{n}].
g_{P}^{l}(\lambda)\triangleq&\lambda B+\sum_{n=1}^{N}a_{n}\max\{0,V_{n}^{l}(\tau_n+D_{n})\},
\end{aligned}
\end{equation}
where $V_{n}^{l}(\tau)$ and $V_{n}^{u}(\tau)$ are defined in \eqref{eq:value-general-perpre-lower} and \eqref{eq:value-general-perpre-upper}. 
Let $g_{P}(\lambda)$ and $\phi_{P}^{*}$ be the dual function and the optimal timely-throughput with perfect prediction in the general case, respectively. We have: 
%For $g_{P}(\lambda)$ and $\phi_{P}^{*}$ defined in Corollary \ref{coro:throughput-simple-nopre}, we have: 
\begin{equation}
g_{P}^{l}(\lambda)\le g_{P}(\lambda)\le g_{P}^{u}(\lambda), \forall \lambda\ge0
\end{equation}
and 
\begin{equation} \label{eq:throughput-general-perpre}
\min_{\lambda}g_{P}^{l}(\lambda)\le\phi^{*}_{P}\le \min_{\lambda}g_{P}^{u}(\lambda).
\end{equation}
\end{coro}
We similarly define $g_{0}^{u}(\lambda)$ and $g_{0}^{l}(\lambda)$  for the zero prediction case. Then, the following theorem characterizes the improvement in weighted timely-throughput from prediction. 
\begin{theorem}\label{thm:improve-general-perpre} %\textcolor{red}{check conditions} 
Let $\lambda_{0}^{*}=\argmin_{\lambda} g_{0}^{l}(\lambda)$ and $\lambda_{P}^{*}=\argmin_{\lambda} g_{P}^{l}(\lambda)$. The weighted timely-throughput improvement satisfies:  
\begin{equation} 
g_{P}^{l}(\lambda_{P}^{*})-g_{0}^{u}(\lambda_{P}^{*})\le \phi^{*}_{P}-\phi^{*}_{0}\le g_{P}^{u}(\lambda_{0}^{*})-g_{0}^{l}(\lambda_{0}^{*}). 
\end{equation}
\end{theorem}
\begin{proof}
Similar to the proof of Theorem \ref{thm:improve-simple-perpre}.
\end{proof}

Theorem \ref{thm:improve-general-perpre} is the counterpart of Theorem \ref{thm:improve-simple-perpre} in the general case, with the main difference that, due to the general Markov dynamics, the bounds can be loose compared to those in Theorem \ref{thm:improve-simple-perpre}. Nonetheless, they are tight when applied to the static setting. 
%\textcolor{red}{talk about order? refer to simulation == } 

\subsubsection{Imperfect prediction} 
We now turn to the imperfect prediction case. Recall that $p_{n}$ is the true-positive probability  and $q_{n}$ is the false-negative probability. 
\begin{theorem}\label{thm:value-general-impre} 
For $1\le n\le N$, $0<\tau<\tau_n$, let $\tilde{V}_{n}^{l}(\tau)=V_{n}^{l}(\tau)$ in \eqref{eq:value-general-perpre-lower}, and $\tilde{V}_{n}^{u}(\tau)=V_{n}^{u}(\tau)$ in  \eqref{eq:value-general-perpre-upper}. 
Then, for $\tau_n\le\tau\le \tau_n+D_{n}$, define 
\begin{eqnarray} 
\hspace{-.5in}&&\tilde{V}_{n}^{l}(\tau)=\max_{e>0}\bigg\{-(1-p_{n})\frac{1-[1-\zeta_{n}(i_{n}^{\min},e)]^{\tau-\tau_n}}{\zeta_{n}(i_{n}^{\min},e)}\lambda e \nonumber \\
\hspace{-.5in}&&\qquad+p_{n}\frac{1-[1-\zeta_{n}(i_{n}^{\min},e)]^{\tau}}{\zeta_{n}(i_{n}^{\min},e)}[-\lambda e+\zeta_{n}(i_{n}^{\min},e)\beta_{n}]\bigg\}, \label{eq:value-general-impre-lower}
%\end{aligned}
\end{eqnarray}
%for $\tau >  \tau_n$, and  define 
and
%\begin{eqnarray} 
%\begin{aligned}
%\hspace{-.2in}&&\tilde{V}_{n}^{u}(\tau) \label{eq:value-general-impre-upper}\\
%\hspace{-.2in}&&=\sum_{z=\tau_n+2}^{\tau}\max_{E}\{-\lambda E+\zeta_{n}(i_{n}^{\max},E)(p_{n}\beta_{n}-\tilde{V}_{n}^{l}(z-1))\}\nonumber\\
%\hspace{-.2in}&&+\max_{E}\{-\lambda E+p_{n}\zeta_{n}(i_{n}^{\max},E)(\beta_{n}-\tilde{V}_{n}^{l}(\tau_n))\}\nonumber\\
%\hspace{-.2in}&&+p_{n}\sum_{z=1}^{\tau_n}\max_{E}\{-\lambda E+\zeta_{n}(i_{n}^{\max},E)(\beta_{n}-\tilde{V}_{n}^{l}(z-1))\},\nonumber
%\end{aligned}
%\end{eqnarray}
\begin{equation}\label{eq:value-general-impre-upper}
\begin{aligned}
\tilde{V}_{n}^{u}(\tau)
=&\sum_{z=\tau_n+1}^{\tau}\max_{e}\{-\lambda e+\zeta_{n}(i_{n}^{\max},e)\\
&\times(p_{n}\beta_{n}-\max\{0,\tilde{V}_{n}^{l}(\tau_n),\tilde{V}_{n}^{l}(z-1)\})\}\\
&+p_{n}\sum_{z=1}^{\tau_n}\max_{e}\{-\lambda e+\zeta_{n}(i_{n}^{\max},e)\\
&\times(\beta_{n}-\max\{0,\tilde{V}_{n}^{l}(z-1)\})\}. 
\end{aligned}
\end{equation}
Notice that $\tilde{V}_{n}^{l}(\tau_n)=p_{n}V_{n}^{l}(\tau_n)$ and $\tilde{V}_{n}^{u}(\tau_n)=p_{n}V_{n}^{u}(\tau_n)$.
%Suppose the prediction window for user $n$ is $D_{n}>0$, the true positive rate is $p_{n}$, and the false negative rate is $q_{n}$. 
Consider a predicted user  $n$ arrival, given a fixed $\lambda$, the value function satisfies: 
\begin{eqnarray} 
&&\hspace{-.2in}V_{n}(0,\tau+w,i)\ge \max\{0,\tilde{V}_{n}^{l}(\tau_n),\tilde{V}_{n}^{l}(\tau_n+w)\} \label{eq:Vn-lower-bdd-general} \\
&&\hspace{-.2in}V_{n}(0,\tau+w,i)\le\min\{p_{n}\beta_{n},\tilde{V}_{n}^{u}(\tau_n+{w})\},\label{eq:Vn-upper-bdd-general}
\end{eqnarray} 
where $0<w\le D_n$.
\end{theorem}
\begin{proof}
See Appendix \ref{apd:value-impre}.
\end{proof}
Interestingly, we will see in the proof that the three terms in $\max\{0,\tilde{V}_{n}^{l}(\tau_n),\tilde{V}_{n}^{l}(\tau_n+w)\}$ in (\ref{eq:Vn-lower-bdd-general}) correspond to not transmitting, transmitting after packet arrival, and proactive transmission.

From Theorem \ref{thm:value-general-impre}, we  have  the following corollary that bounds the optimal timely-throughput. 
\begin{coro}\label{coro:throughput-general-impre}
Define
\begin{eqnarray} 
%g_{P}^{u}(\lambda)\triangleq&\lambda B+\sum_{n=1}^{N}a_{n}\min(\beta_{n},\\
%&(\tau_n+D_{n})\max_{E}\{-\lambda E+\zeta_{n}(i_{n}^{\max},E)\beta_{n}\})
&g_{I}^{u}(\lambda)\triangleq\lambda B+\sum_{n=1}^{N}\big\{\tilde{a}_{n}\min[p_{n}\beta_{n},\tilde{V}_{n}^{u}(\tau_n+D_{n})]\\
&+(a_{\max}-\tilde{a}_{n})q_{n}\min[\beta_{n},V_{n}^{u}(\tau_n)]\big\},\nonumber 
\end{eqnarray}
and 
\begin{eqnarray}
%g_{P}^{l}(\lambda)\triangleq&\lambda B+\sum_{n=1}^{N}a_{n}\max_{E}\frac{1-[1-\zeta_{n}(i_{n}^{\min},E)]^{\tau_n+D_{n}}}{\zeta_{n}(i_{n}^{\min},E)} \\
%&\times[-\lambda E+\zeta_{n}(i_{n}^{\min},E)\beta_{n}].
\hspace{-.15in}&g_{I}^{l}(\lambda)\triangleq\lambda B+\sum_{n=1}^{N}\big\{\tilde{a}_{n}\max[0, \tilde{V}_{n}^{l}(\tau_n), \tilde{V}_{n}^{l}(\tau_n+D_{n})]\\
\hspace{-.15in}&+(a_{\max}-\tilde{a}_{n})q_{n}\max[0,V_{n}^{l}(\tau_n)]\big\}.\nonumber
\end{eqnarray}
Let $g_{I}(\lambda)$ and $\phi_{I}^{*}$ be the dual function and the optimal timely-throughput in the imperfect prediction case, respectively. We have: 
%defined the same as in Corollary \ref{coro:throughput-simple-impre}, we have: 
\begin{equation}
g_{I}^{l}(\lambda)\le g_{I}(\lambda)\le g_{I}^{u}(\lambda), \forall \lambda\ge0
\end{equation}
and 
\begin{equation} \label{eq:throughput-general-impre}
\min_{\lambda}g_{I}^{l}(\lambda)\le\phi^{*}_{I}\le \min_{\lambda}g_{I}^{u}(\lambda).
\end{equation}
\end{coro}
From the above, we can now bound the improvement in weighted timely-throughput with imperfect prediction. 
\begin{theorem}\label{thm:improve-general-impre}
Suppose $\lambda_{I}^{*}=\argmin_{\lambda} g_{I}^{l}(\lambda)$, then the weighted timely-throughput improvement satisfies:  
\begin{equation}
g_{I}^{l}(\lambda_{I}^{*})-g_{0}^{u}(\lambda_{I}^{*})\le \phi^{*}_{I}-\phi^{*}_{0}\le g_{I}^{u}(\lambda_{0}^{*})-g_{0}^{l}(\lambda_{0}^{*}).
\end{equation}
\end{theorem}
\begin{proof}
Similar to the proof of Theorem \ref{thm:improve-simple-perpre}.
\end{proof}

Although it is hard to exactly characterize the timely-throughput improvement and how it depends on $p_n$  and $q_n$ in this general case (due to the general system dynamics),  Theorems \ref{thm:improve-general-perpre} and  \ref{thm:improve-general-impre}  provide useful upper and lower bounds. 
We will also see in the simulation section that the performance in the general case is  similar to that   in the static case as characterized in Theorems \ref{thm:improve-simple-perpre} and \ref{thm:improve-simple-impre}. 

%conduct extensive simulations in Section \ref{sec:simulation} to demonstrate that th

%Theorems \ref{thm:improve-general-perpre} and  \ref{thm:improve-general-impre} provide characterization of the timely-throughput improves for the perfect prediction and imperfect prediction cases. While the general scaling order is hard to exactly quantify as in the deterministic case, we observe in simulation that the timely-throughput improvement is roughly of the same 
%\textcolor{red}{== to be done ==}

%\textcolor{red}{==  2. throughput improvement with/without prediction error, iid channel ==}  

%\textcolor{red}{==   3.  sensitivity of the results to estimation error of $p_n, q_n$ and $p_n$ == }

%\textcolor{red}{=== Optimality with capacity constraint ===} 

\subsection{Capacity Constraint}
In our model, we have assumed that there is no transmission capacity constraint to facilitate analysis. Here we discuss what happens if a hard constraint is imposed in \eqref{eq:problem}, i.e., 
%If we impose a capacity constraint to the server  i.e.,
\begin{eqnarray}
\sum_{n=1}^{N}c_{n}(t)\le C, \forall t,
\end{eqnarray}
where $c_{n}(t)$ is the number of user $n$ packets transmitted at time $t$.

In this case, the problem is significantly different as the scheduling decisions of packets are dependent.  
We can construct a truncated policy $\hat{\pi}^{*}$ based on the optimal policy $\pi^*$ derived above. 
\begin{enumerate}[(a)] 
\item Solve problem \eqref{eq:problem} to get the optimal policy $\pi^{*}$. 
\item At slot $t$, define $\Omega_{\pi^{*}}(t)$ as the set of packets that are scheduled to be transmitted under $\pi^{*}$. 
If $|\Omega_{\pi^{*}}(t)|\le C$, then the scheduling decisions of $\hat{\pi}^{*}$ are the same as $\pi^{*}$. Otherwise if $|\Omega_{\pi^{*}}(t)|> C$, randomly select a subset of packets $\Omega_{\hat{\pi}^{*}}(t)\subset\Omega_{\pi^{*}}(t)$ with $|\Omega_{\hat{\pi}^{*}}(t)|=C$ to transmit at the same levels as decided by $\pi^{*}$, and discard packets in $\Omega_{\pi^{*}}(t)-\Omega_{\hat{\pi}^{*}}(t)$.
\end{enumerate}
The truncated policy $\hat{\pi}^{*}$ can be shown to be asymptotically optimal with a similar argument as that in \cite{singh2016throughput}. The proof is similar  and is omitted here.

\section{Simulation Results}  \label{sec:simulation}
We present  simulation results of the optimal scheduling policy in this section. Our simulation is conducted for the system in Fig. \ref{fig:system} with $N=4$ users. 
% 
%We distinguish the users by their different distance to the server, and set the distance vector as $\bm{d}=(d_{1},\dots,d_{4})=(1.1,1.2,1.3,1.4)$. 
% 
The arrival rates of packets are given by $(a_{1},\dots,a_{4})=(0.7,0.6,0.4,0.3)$ with $a_{\max}=1$ and the  deadlines are given by $(\tau_{1},\dots,\tau_{4})=(2,3,4,5)$. 
We set the reward weight vector to be $\bm{\beta}=(3,1,2,4)$. Each channel has  $K=4$ states $(s_{1}, \dots, s_{4})=(1,2,3,4)$, each representing a noise level. We assume the channel state transition matrix is the same for all users and is given by: 
\begin{equation*}
P=\left[
\begin{matrix}
0.4 & 0.3 & 0.2 & 0.1 \\
0.25 & 0.3 & 0.25 & 0.2 \\
0.2 & 0.25 & 0.3 & 0.25 \\
0.1 & 0.2 & 0.3 & 0.4
\end{matrix}
\right] 
\end{equation*}
The set of resource levels $\mathcal{E}$ is a discrete set of the form $\{0.0001z, z=0,1, ..., 6\times10^4\}$. The  average resource expenditure budget is $B=6$. The probability that a transmission for user $n$ is successful under state $s_{i}$ and resource level $e$ is given by: 
\begin{equation}
\zeta_{n}(i,e)=\frac{2}{1+e^{-2\frac{e}{d_{n}^{3}s_{i}}}}-1,
\end{equation}
where $\bm{d}=(d_{1}, \dots, d_{4})=(1.1,1.2,1.3,1.4)$ and $d_n$ denotes the distance between user $n$ and the server. This setting models a wireless downlink system. %\cite{tse2005fundamentals}. %\textcolor{red}{give a reference here. David Tse's communication book?}
The prediction window sizes are the same for all users, i.e., $D_{n}=2, n=1,\dots,4.$ For the imperfect prediction case, the true-positive rates and false-negative rates are $\bm{p}=(0.8,0.8,0.8,0.8)$ and $\bm{q}=(0.2,0.1,0.1,0.2)$. 

\subsection{Optimal Policy} \label{sec:simulation-policy}
We start with the optimal actions for different users under different states. 
% 
%We take User $2$ and User $4$ as representative. 
% 
Table \ref{tab:noprediction}  and Table \ref{tab:perprediction} show the optimal scheduling decisions for User $2$ ($\tau_2=3$) and User $4$ ($\tau_4=5$), i.e., $e_{n}^{*}(0,\tau,i)$, in the zero prediction and perfect prediction cases. 
From the results, we can verify the monotonicity of the optimal policy, as shown in Theorem \ref{theorem:monotone}. 
We also see that the resource offered by the server for User $4$ is larger than that under the same state for User $2$. This is intuitive since $\beta_{4}>\beta_{2}$. 
\begin{table}[ht]
\centering
\subtable[User 2]{
\begin{tabular}{|c|c|c|c|} 
\hline
\diagbox{S(t)}{$\tau$}  &3 & 2 & 1 \\
\hline
$s_{1}$ & 1.3915 &1.4674 & 1.6328 \\
\hline
$s_{2}$ & 0.0 & 0.2762 &1.0554 \\
\hline
$s_{3}$ & 0.0 & 0.0 & 0.0 \\
\hline
$s_{4}$ & 0.0 & 0.0 & 0.0\\
\hline
\end{tabular}
}
\subtable[User 4]{        
\begin{tabular}{|c|c|c|c|c|c|} 
\hline
\diagbox{S(t)}{$\tau$}  &5 &4 &3 & 2 & 1 \\
\hline
$s_{1}$ & 2.6024 & 2.7488 & 2.9284 & 3.2269 & 4.1129  \\
\hline
$s_{2}$ & 2.3917 & 2.9635 & 3.5677 & 4.3658 & 6.0 \\
\hline
$s_{3}$ & 0.0 & 0.0 & 1.6809 & 3.9651 & 6.0  \\
\hline
$s_{4}$ & 0.0 & 0.0 & 0.0 & 1.9522 & 6.0\\
\hline
\end{tabular}
}
\caption{Optimal Scheduling Decisions $e_{n}^{*}(0,\tau,i)$ (Zero Prediction)}
\vspace{-.2in}
\label{tab:noprediction}
\end{table}

\begin{table}[ht]
\centering
\subtable[User 2]{
\begin{tabular}{|c|c|c|c|c|c|} 
\hline
\diagbox{S(t)}{$\tau$}  & 3+2 & 3+1 &3 & 2 & 1 \\
\hline
$s_{1}$ & 1.3651 & 1.4248 & 1.4906 & 1.5797 & 1.7865    \\
\hline
$s_{2}$ & 0.0 & 0.0 & 0.6726 & 1.1571 & 1.6759  \\
\hline
$s_{3}$ & 0.0 & 0.0 & 0.0 & 0.0 & 0.0  \\
\hline
$s_{4}$ & 0.0 & 0.0 & 0.0 & 0.0 & 0.0 \\
\hline
\end{tabular}
}
\subtable[User 4]{        
\begin{tabular}{|c|c|c|c|c|c|c|c|} 
\hline
%\diagbox{S(t)}{$\tau$} & 5+2 & 5+1 & 5 & 4 & 3 & 2 & 1 \\
\diagbox{S(t)}{$\tau$}  & 5+2 & 5+1 & 5 & 4 & 3 & 2 & 1 \\
\hline
$s_{1}$& 2.3912 & 2.5051 & 2.6355 & 2.7873 & 2.9807 & 3.3168 & 4.3156 \\
\hline
$s_{2}$ & 1.2882 & 1.948 & 2.5292 & 3.1024 & 3.7309 & 4.6111 & 6.0   \\
\hline
$s_{3}$ & 0.0 & 0.0 & 0.0 & 0.0 & 2.315 & 4.5133 & 6.0  \\
\hline
$s_{4}$ & 0.0 & 0.0 & 0.0 & 0.0 & 0.0 & 3.4852 & 6.0 \\
\hline
\end{tabular}
}
\caption{Optimal Scheduling Decisions $e_{n}^{*}(0,\tau,i)$ (Perfect Prediction)}
\label{tab:perprediction}
\vspace{-.2in}
\end{table}

Table \ref{tab:imprediction} shows the optimal policy in the imperfect prediction case. From the scheduling decisions for User $2$ under channel state $s_{1}$, we  see that \emph{the monotonicity property actually does not hold when there is prediction error}, i.e.,  resource expenditure for state $s_1$ and $\tau=3+1$ is smaller than that for state $s_1$ and $\tau=3+2$. 
Another interesting observation is that the resource allocated in the imperfect prediction case is less than that for the same state in the perfect prediction case. This is because resource may be wasted when prediction is imperfect. Thus, the server is more conservative in resource allocation.  
\begin{table}[ht]
\centering
\subtable[User 2]{
\begin{tabular}{|c|c|c|c|c|c|} 
\hline
\diagbox{S(t)}{$\tau$}  & 3+2 & 3+1 &3 & 2 & 1 \\
\hline
$s_{1}$ & \textbf{0.8382} & \textbf{0.8269} & 1.2468 & 1.313 & 1.4435   \\
\hline
$s_{2}$ & 0.0 & 0.0 & 0.0 & 0.0 & 0.0 \\
\hline
$s_{3}$ & 0.0 & 0.0 & 0.0 & 0.0 & 0.0  \\
\hline
$s_{4}$ & 0.0 & 0.0 & 0.0 & 0.0 & 0.0 \\
\hline
\end{tabular}
}
\subtable[User 4]{        
\begin{tabular}{|c|c|c|c|c|c|c|c|} 
\hline
%\diagbox{S(t)}{$\tau$} & 5+2 & 5+1 & 5 & 4 & 3 & 2 & 1 \\
\diagbox{S(t)}{$\tau$}  & 5+2 & 5+1 & 5 & 4 & 3 & 2 & 1 \\
\hline
$s_{1}$& 1.9575 & 1.9758 & 2.5562 & 2.6954 & 2.8624 & 3.1235 & 3.8744 \\
\hline
$s_{2}$ & 0.0 & 0.0 & 2.1895 & 2.7631& 3.3489 & 4.0761 & 5.4601  \\
\hline
$s_{3}$ & 0.0 & 0.0 & 0.0 & 0.0 & 0.0 & 3.2527 & 5.7612  \\
\hline
$s_{4}$ & 0.0 & 0.0 & 0.0 & 0.0 & 0.0 & 0.0 & 4.5436\\
\hline
\end{tabular}
}
\caption{Optimal Scheduling Decisions $e_{n}^{*}(0,\tau,i)$ (Imperfect Prediction)}
\label{tab:imprediction}
\vspace{-.2in}
\end{table} 

\subsection{Influence of Parameters}
Next we investigate the influence of system parameters. 
% 
%To reduce computation time, in this subsection we change the energy-level set to $\mathcal{E}=\{0, 1, \dots, 6\}$, while keeping other settings unchanged. %\textcolor{red}{The results with prediction refers to those in the imperfect prediction case.???} 
% 
Fig. \ref{fig:throughput-tau} shows how the weighted timely-throughput changes 
%with deadline and prediction window size,  in the case 
when all users have the same deadline going from  $2$ to $7$ ($D_n=2$). 
We see that  the timely-throughputs in both cases, with or without prediction,  increase with the deadline, and the throughput improvement decreases as $\tau_n$ becomes larger. This is expected, as a large deadline gives more flexibility to scheduling. Thus, the marginal benefit of prediction decreases. 
Moreover, we see that the throughput with prediction is always higher than that without prediction, demonstrating the effectiveness of prediction. 
\begin{figure}[ht]
\centering
\vspace{-.1in}
\includegraphics[width=3.4in, height=2.8in]{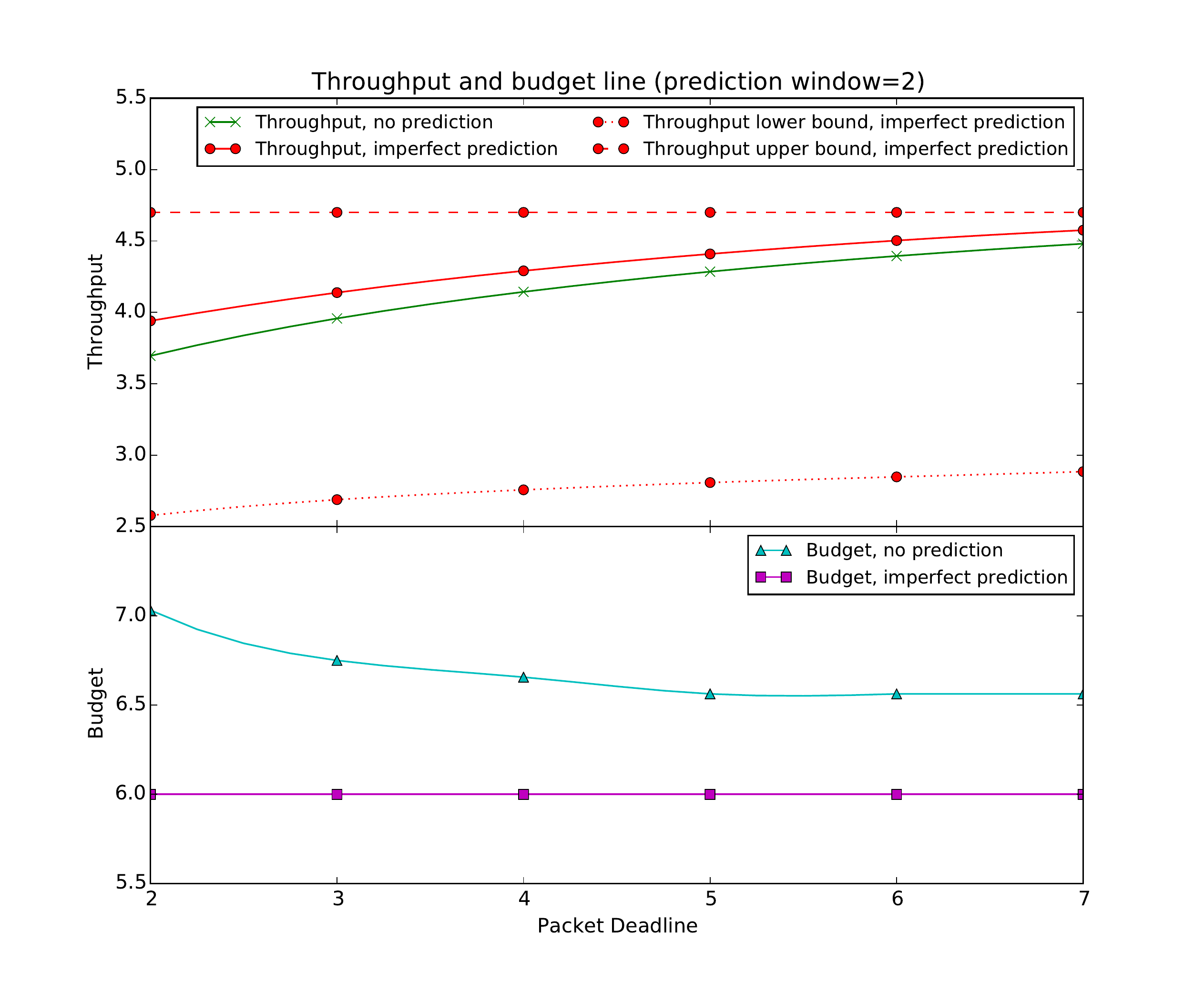}
\vspace{-.1in}
\caption{The top plot shows that timely-throughput increases in packet deadline. In the bottom plot, we show the resource budget needed in the zero prediction case to achieve the same throughput as in the imperfect prediction case. With prediction, the overall resource consumption is set as $6$. However, without prediction, we need a significantly higher resource consumption to achieve the same throughput.} 
%\textcolor{red}{=not clear=}} 
\label{fig:throughput-tau}
\vspace{-.1in}
\end{figure}

Fig. \ref{fig:throughput-window} shows how the optimal weighted throughput changes with prediction power. From the results, we see that although prediction is imperfect, if used properly, it can still significantly improve timely-throughput. The throughput bounds in Fig. \ref{fig:throughput-tau} and Fig. \ref{fig:throughput-window} appear loose due to the complicated Markov dynamics. Fig. \ref{fig:bound}, on the other hand, shows that the bounds are tight when there is a single channel state.
%despite potential  errors in prediction,  the throughput increases in the prediction window. 
\begin{figure}[ht]
\centering
\vspace{-.1in}
\includegraphics[width=3.4in, height=2.8in]{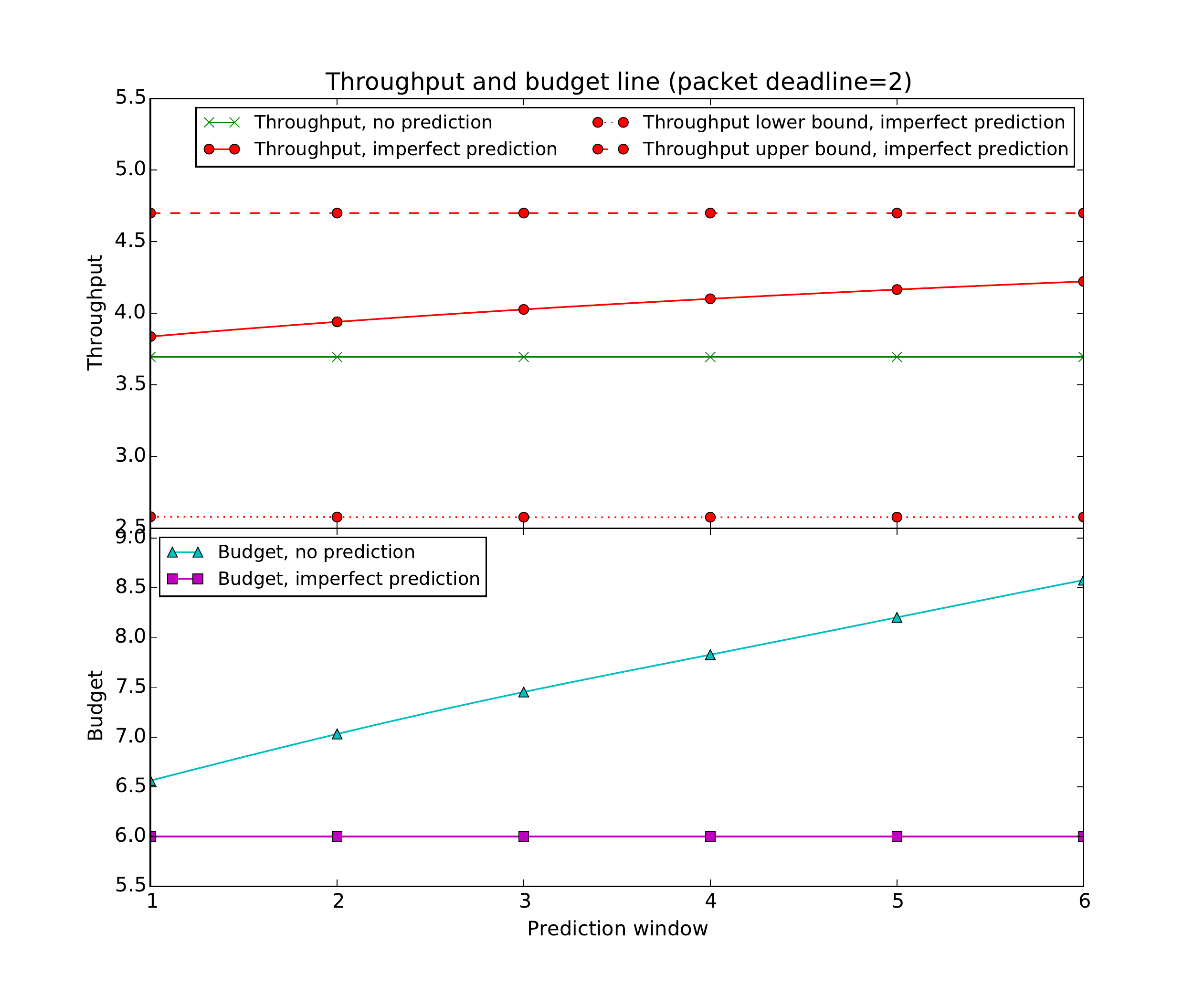}
\vspace{-.1in}
\caption{The top plot shows how timely-throughput changes with the prediction window size. The bottom plot  shows the resource budget needed for achieving the same throughputs as with prediction. Without prediction, one needs a significantly higher resource consumption in order to achieve the same performance.} 
\label{fig:throughput-window} 
\vspace{-.05in}
\end{figure}

\begin{figure}[ht]
\centering
\vspace{-.05in}
\includegraphics[width=1.7in, height=1.4in]{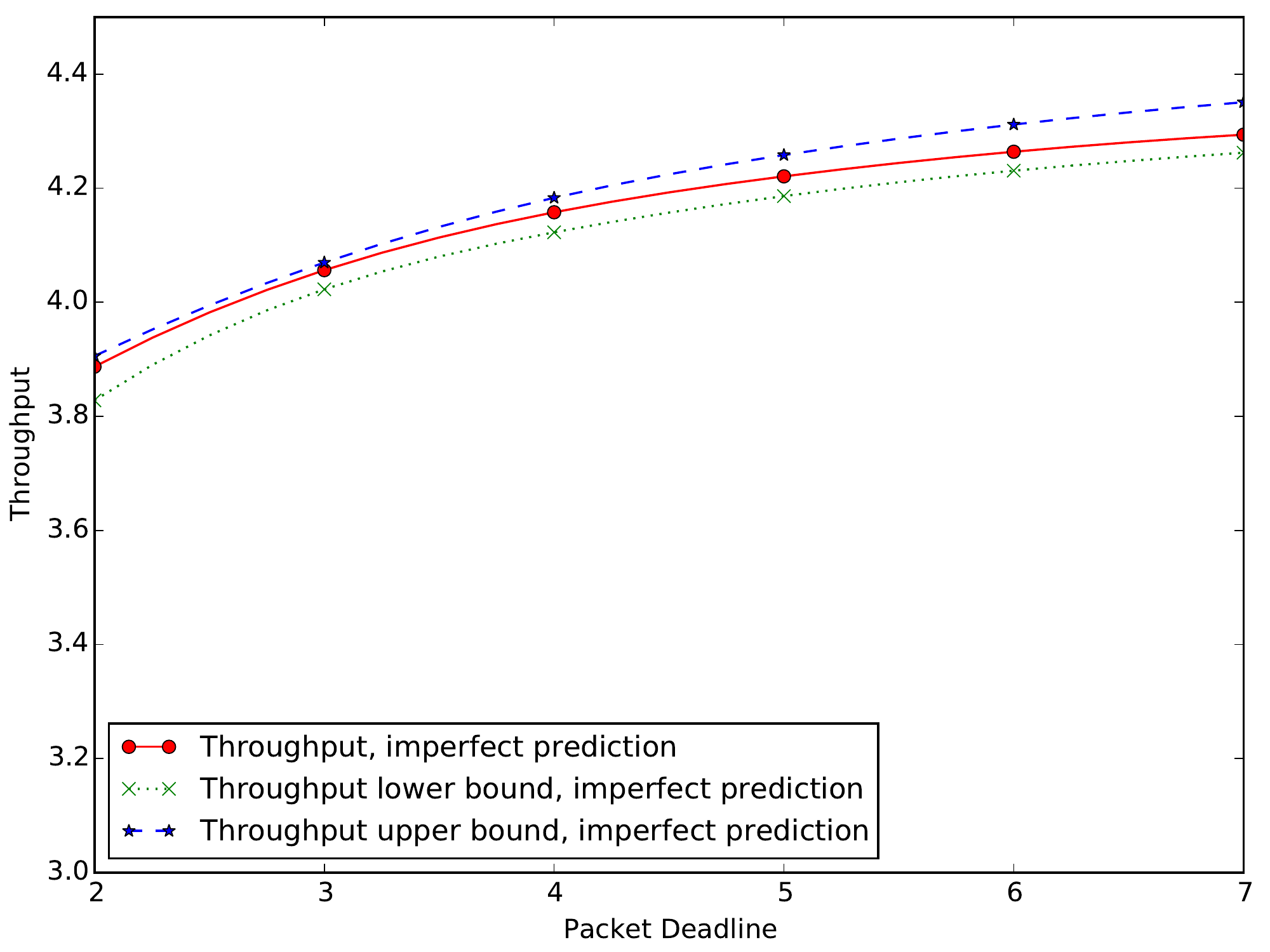}
\includegraphics[width=1.7in, height=1.4in]{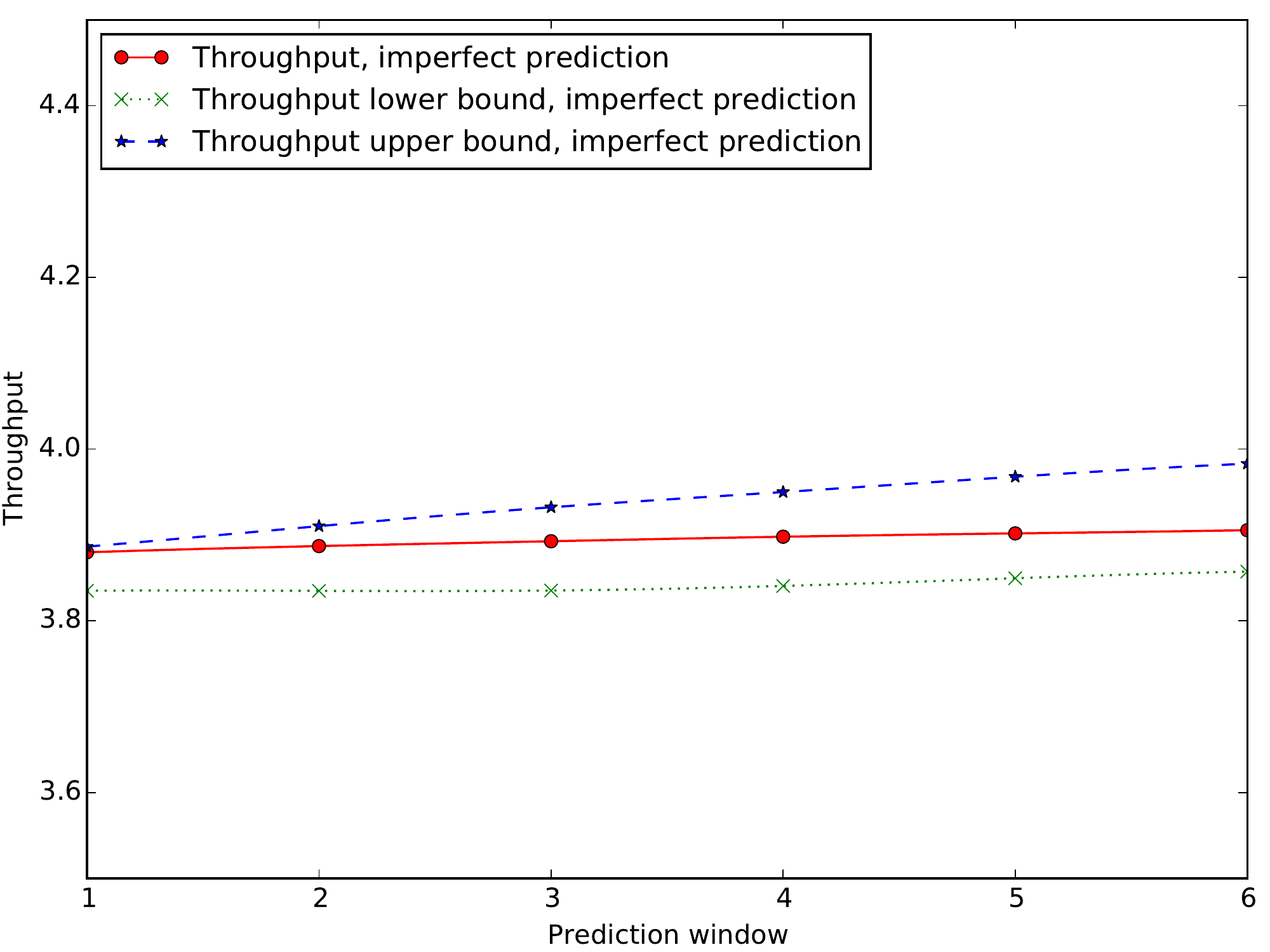}
\vspace{-.1in}
\caption{Timely-throughput and the corresponding upper and lower bounds. The set of channel states is $\mathcal{S}=\{s_{2}\}$. We see that the bounds are tight when there is only a single channel state.} 
\label{fig:bound}
\vspace{-.1in}
\end{figure}

%Fig. \ref{fig:throughput-pq} shows the relationship between the optimal timely-throughput and the prediction accuracy. 
% 
Fig. \ref{fig:throughput-surface} shows the relationship between the optimal timely-throughput and  prediction accuracy, which is actually complicated.
We also show in Fig. \ref{fig:throughput-pq} results with fixed $\bm{q}$ and $\bm{p}$, respectively, to show how timely throughput changes with the other. 
%For ease of exposition, we take two cross-sections from Fig. \ref{fig:throughput-surface} and obtain Fig. \ref{fig:throughput-pq}. 
In the left plot, the false-negative rate vector is set to $\bm{q}=\bm{0}$, and all users have the same true-positive rate increasing from $0.75$ to $0.95$. In the right plot, the true-positive rate vector is set as $\bm{p}=(0.8,0.8,0.8,0.8)$ and all users have the same false-negative rate changing from $0.05$ to $0.25$. 
Other settings are kept unchanged as in Section \ref{sec:simulation-policy}. 
From the results, we  see that the timely-throughput is increasing in the true-positive rate and decreasing in the false-negative rate for this stochastic case (this is proven for the static scenario in Theorem \ref{throughput-simple-p-q}).   
%\textcolor{red}{= this is the general case?=}
\begin{figure}[ht]
\centering
%\vspace{-.1in}
\includegraphics[width=3in]{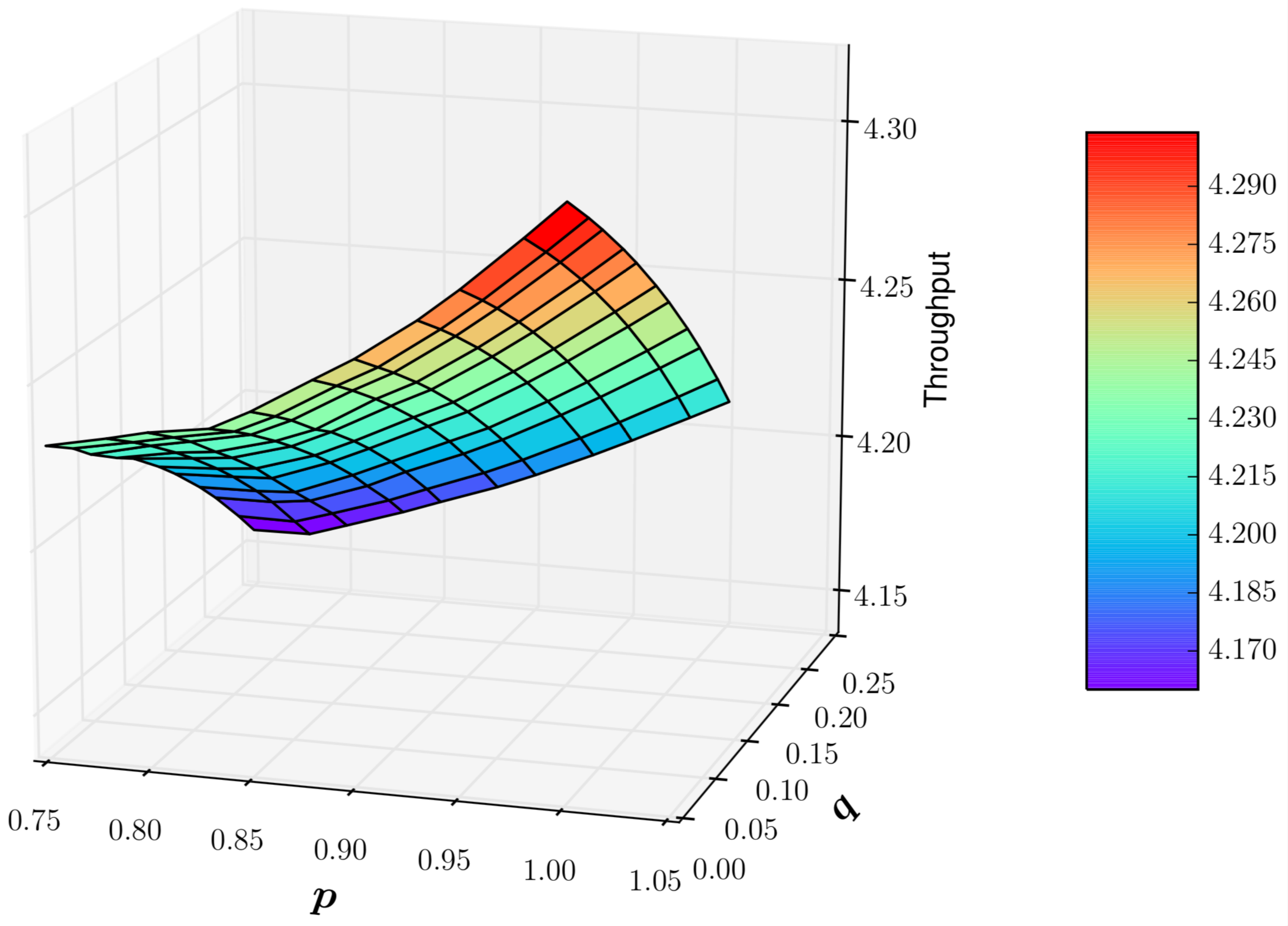}
\vspace{-.1in}
\caption{Timely-throughput as a function of $p$ and $q$. The relationship between timely-throughput and general $(p,q)$ pair is complicated.}%\textcolor{red}{=change=}}
\label{fig:throughput-surface}
%\vspace{-.1in}
\end{figure}

\begin{figure}[ht]
\centering
\vspace{-.1in}
\includegraphics[width=1.7in, height=1.5in]{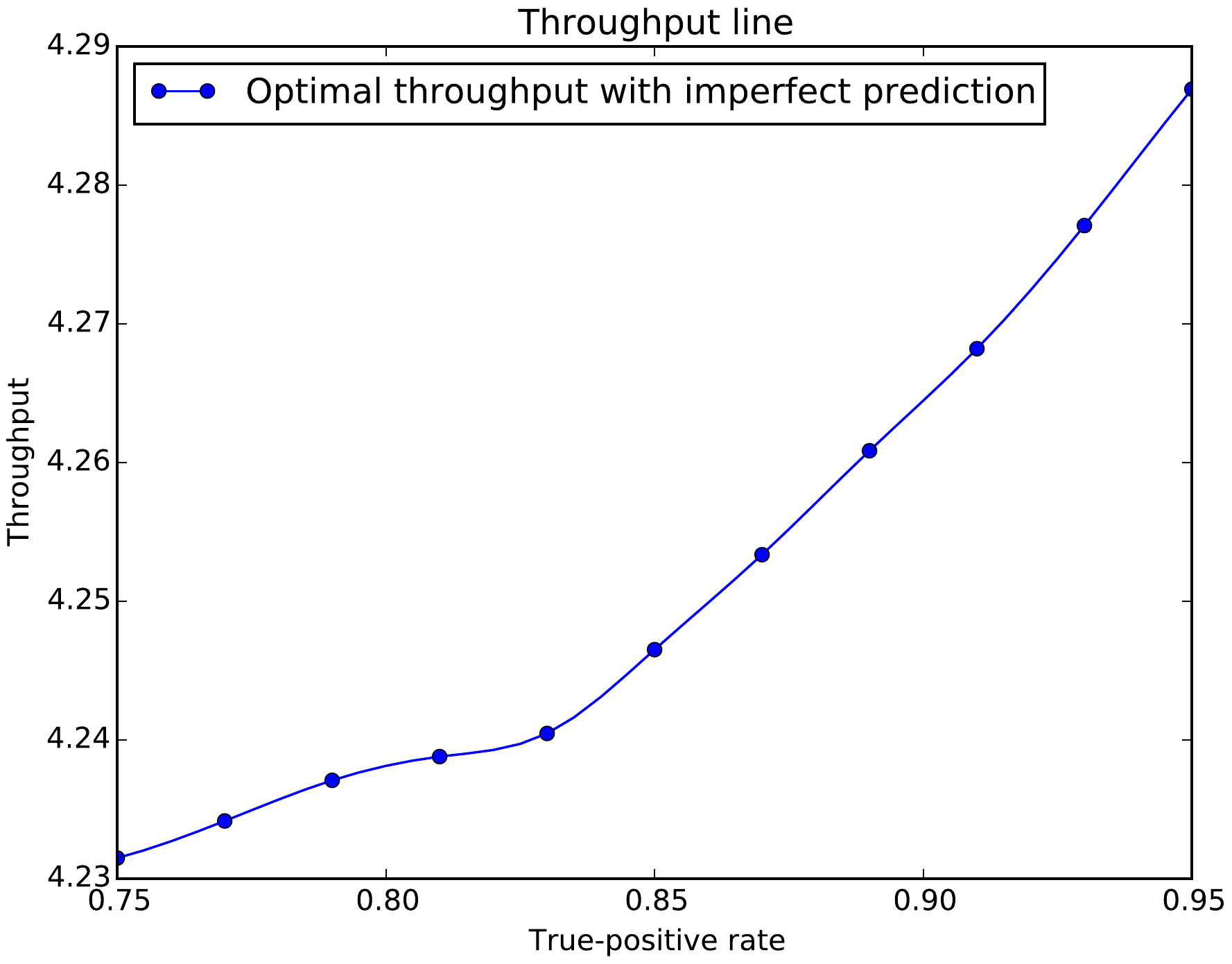}
\includegraphics[width=1.7in, height=1.5in]{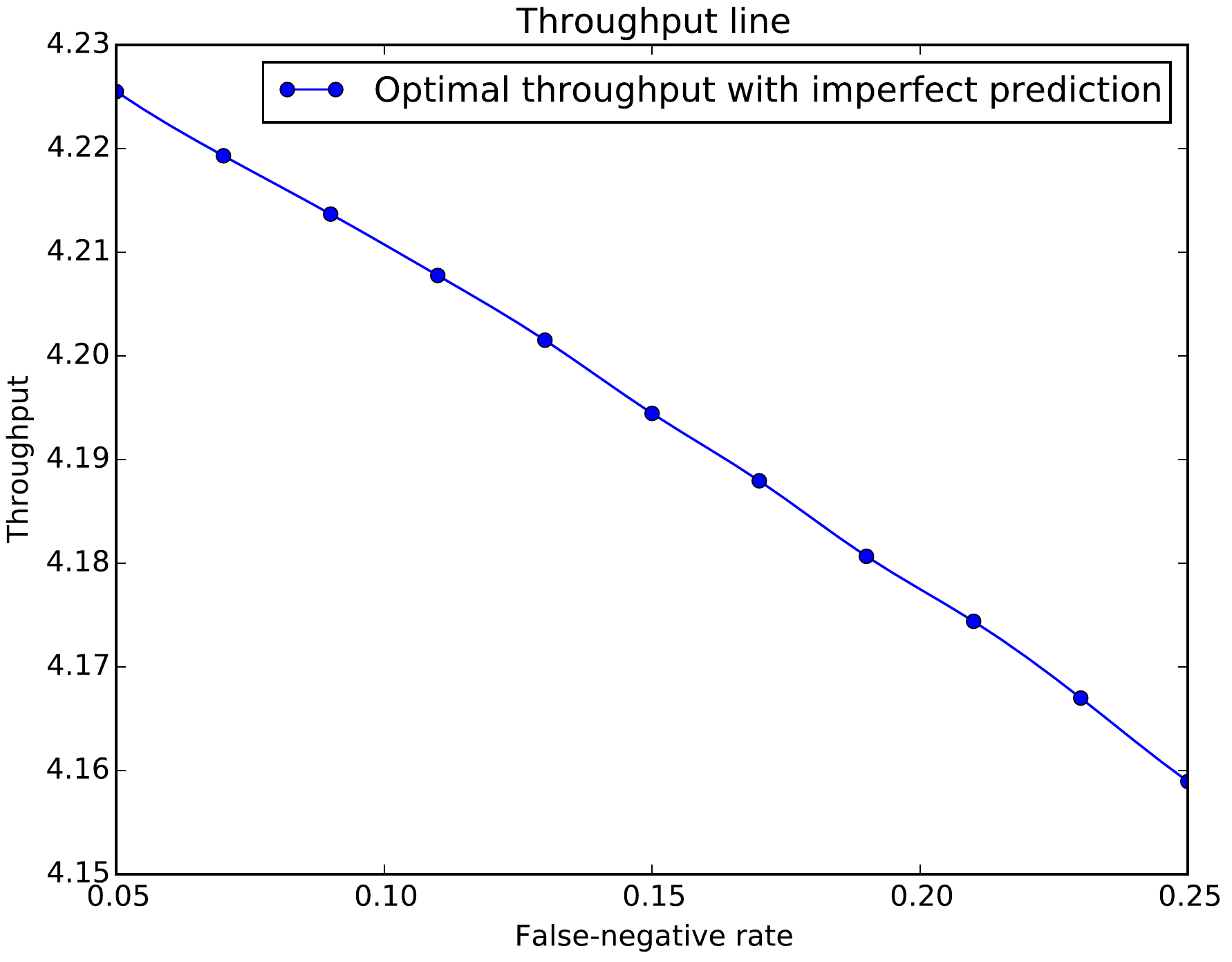}
\vspace{-.1in}
\caption{Timely-throughput change with $p$ and $q$. Left plot: $\bm{q}=\bm{0}$ and $p$ increases from $0.75$ to $0.95$. Right plot: $\bm{p}=(0.8,0.8,0.8,0.8)$ and $q$ increases from $0.05$ to $0.25$.} 
\label{fig:throughput-pq}
\vspace{-.1in}
\end{figure}

%\begin{figure}[ht]
%\centering
%\vspace{-.1in}
%\includegraphics[width=3in, height=2in]{throughput-q.png}
%\vspace{-.1in}
%\caption{Timely-throughput decreases in $q$.}
%\label{fig:throughput-q}
%\vspace{-.1in}
%\end{figure} 

\section{Conclusion}\label{sec:conclusion} 
In this paper, we investigate the fundamental benefit of predictive scheduling in improving timely-throughput in a general stochastic single-server multi-user system. 
%, and investigate the fundamental benefit of predictive scheduling in improving timely-throughput, which is the rate of packets that are delivered to destinations before their deadlines. 
% 
Based on an error rate-based prediction model, we first derive a Markov decision process (MDP) solution for optimizing the timely-throughput objective, subject to an average resource consumption constraint. 
We then explicitly characterize the optimal scheduling policy and quantify the timely-throughput improvement due to predictive-service. Extensive simulations are conducted to validate our theoretical results.   
Our results provide novel insights into how prediction and system parameters impact  performance and provide useful guidelines for designing future predictive control algorithms.

\bibliographystyle{IEEEtran}
\bibliography{harddeadline-ref.bib}

\appendices
\section{Proof of Theorem \ref{thm:policy-simple-nopre}} \label{apd:policy-simple-nopre}
Here we prove Theorem \ref{thm:policy-simple-nopre}. Recall that the state of a packet is described by $(r, \tau)$ in this case, where $r\in\{0, 1\}$ indicates whether it is delivered and $0\leq\tau\leq \tau_n$ denotes the remaining time until deadline. 
\begin{proof} (Theorem \ref{thm:policy-simple-nopre})
In the static scenario, the Bellman equations in \eqref{eq:dp} to \eqref{eq:dp2} reduce to
\begin{equation}\label{eq:dp-simple}
\begin{aligned}
V_{n}(0,\tau)=&\max\big\{V_{n}(0,\tau-1),-\lambda +\zeta_{n}V_{n}(1,\tau-1)\\
&+(1-\zeta_{n})V_{n}(0,\tau-1)\big\},0<\tau\le \tau_n+D_{n}, \\
V_{n}(0,0)=&0,\forall n, \\
V_{n}(1,\tau)=&\beta_{n},\forall n, 0\le\tau<\tau_n+D_{n}.
\end{aligned}
\end{equation}
%From \eqref{eq:dp-simple}, now we prove the former part of Theorem \ref{thm:policy-simple-nopre}, i.e., if $\zeta_{n}\beta_{n}>\lambda$, then the optimal scheduling is to transmit packets from user $n$ at each time-slot, until they are successfully delivered or outdated, and the value function is $V_{n}(0,\tau)=\frac{1-(1-\zeta_{n})^{\tau}}{\zeta_{n}}(-\lambda+\zeta_{n}\beta_{n}), 0<\tau\le \tau_n.$
%We now prove Theorem \ref{thm:policy-simple-nopre} via induction. 
% 
We now prove the first part of Theorem \ref{thm:policy-simple-nopre} via induction, i.e., if $\zeta_{n}\beta_{n}>\lambda$, then the optimal policy is to transmit the packet at every time-slot, until it is either successfully delivered to the destination or becomes outdated, and the value function is $V_{n}(0,\tau)=\frac{1-(1-\zeta_{n})^{\tau}}{\zeta_{n}}(-\lambda+\zeta_{n}\beta_{n}), 0<\tau\le \tau_n+D_{n}.$ 

First, we show that this statement holds for $\tau=1$. To see this, note that solving the one-step equation of \eqref{eq:dp-simple} gives:  
\begin{eqnarray}
V_{n}(0,1)=\max\{0,-\lambda+\zeta_{n}\beta_{n}\}=-\lambda+\zeta_{n}\beta_{n}. 
\end{eqnarray}
Thus, if  $-\lambda+\zeta_{n}\beta_{n}>0$, the optimal scheduling decision at state $(0,1)$ is to transmit the packet. %Otherwise, the server should remain idle. 

Now suppose the statement holds for $\tau=1,\dots,t$, we will show that it holds for $\tau=t+1$. From \eqref{eq:dp-simple}, we have: 
\begin{equation*}
\begin{aligned}
V_{n}(0,t+1)=\max\big\{&V_{n}(0,t),-\lambda +\zeta_{n}\beta_{n}\\
&\qquad\qquad +(1-\zeta_{n})V_{n}(0,t)\big\}\\
=\max\big\{&\frac{1-(1-\zeta_{n})^{t}}{\zeta_{n}}(-\lambda+\zeta_{n}\beta_{n}),\\
&\qquad\frac{1-(1-\zeta_{n})^{t+1}}{\zeta_{n}}(-\lambda+\zeta_{n}\beta_{n})\}. 
\end{aligned}
\end{equation*}
Thus, if $-\lambda+\zeta_{n}\beta_{n}>0$,  $V_{n}(0,t+1)=\frac{1-(1-\zeta_{n})^{t+1}}{\zeta_{n}}(-\lambda+\zeta_{n}\beta_{n})$, and the optimal scheduling decision at state $(0,t+1)$ is to transmit the packet. %Otherwise, the server will remain idle. 
The proof of the second part goes the same. This completes the proof of Theorem \ref{thm:policy-simple-nopre}. %The proof of the latter part is the same and is omitted here.
\end{proof}

\section{Proof of Theorem \ref{thm:policy-simple-impre}} \label{apd:policy-simple-impre}
Here we prove Theorem \ref{thm:policy-simple-impre}.
\begin{proof} (Theorem \ref{thm:policy-simple-impre})
First note that the Bellman equations in  \eqref{eq:dp-aug}  to \eqref{eq:dp-aug-4} reduce to the following. 
\begin{eqnarray}
\hspace{-.1in}V_{n}(0,\tau)&=&\max\big\{V_{n}(0,\tau-1),-\lambda +\zeta_{n}V_{n}(1,\tau-1)\nonumber\\
\hspace{-.1in}&&\qquad\quad +(1-\zeta_{n})V_{n}(0,\tau-1)\big\},\nonumber\\
\hspace{-.1in}&&\qquad\quad0<\tau\le \tau_n+D_{n}, \tau\ne \tau_n+1\nonumber\\
\hspace{-.1in}V_{n}(0,\tau_n+1)&=&\max\big\{p_{n}V_{n}(0,\tau_n),-\lambda +\zeta_{n}V_{n}(1,\tau_n)\nonumber\\
\hspace{-.1in}&&\qquad\quad+(1-\zeta_{n})p_{n}V_{n}(0,\tau_n)\big\}, \label{eq:dp-aug-simple}\\
\hspace{-.1in}V_{n}(0,0)&=&\,0,\forall n, \nonumber\\
\hspace{-.1in}V_{n}(1,\tau)&=&\,\beta_{n},\forall n, 0\le\tau<\tau_n,\nonumber\\
\hspace{-.1in}V_{n}(1,\tau)&=&\,p_{n}\beta_{n},\forall n, \tau_n\le\tau<\tau_n+D_{n}.\nonumber
\end{eqnarray}

We start from Part (A), i.e., $\zeta_{n}\beta_{n}>\lambda$. Consider $p_{n}> c_{n} =\frac{\lambda}{(-\lambda+\zeta_{n}\beta_{n})(1-\zeta_{n})^{\tau_n}+\lambda}$. 
From Theorem \ref{thm:policy-simple-nopre}, we know that $V_{n}(0,\tau_n)=\frac{1-(1-\zeta_{n})^{\tau_n}}{\zeta_{n}}(-\lambda+\zeta_{n}\beta_{n})$. We want to prove the statement by induction on $w$. 

First, the statement holds for $w=1$. To see this, note that  
\begin{eqnarray*}
\hspace{-.2in}&&V_{n}(0,\tau_n+1)=\max\big\{p_{n}V_{n}(0,\tau_n),\\
\hspace{-.2in}&&\qquad\qquad\qquad\quad  -\lambda +\zeta_{n}V_{n}(1,\tau_n)+(1-\zeta_{n})p_{n}V_{n}(0,\tau_n)\big\}. 
\end{eqnarray*}
Since $p_{n}> c_{n}$, one can verify that $-\lambda +\zeta_{n}V_{n}(1,\tau_n)+(1-\zeta_{n})p_{n}V_{n}(0,\tau_n)>p_{n}V_{n}(0,\tau_n)$. Therefore, 
\begin{eqnarray*}
&&V_{n}(0,\tau_n+1)\\
&=&-\lambda +\zeta_{n}V_{n}(1,\tau_n)+(1-\zeta_{n})p_{n}V_{n}(0,\tau_n) \\
&=&-(1-p_{n})\lambda+p_{n}\frac{1-(1-\zeta_{n})^{\tau_n+1}}{\zeta_{n}}(-\lambda+\zeta_{n}\beta_{n}), 
\end{eqnarray*}
%$V_{n}(0,\tau_n+1)=-\lambda +\zeta_{n}V_{n}(1,\tau_n)+(1-\zeta_{n})p_{n}V_{n}=-(1-p_{n})\lambda+p_{n}\frac{1-(1-\zeta_{n})^{\tau_n+1}}{\zeta_{n}}(-\lambda+\zeta_{n}\beta_{n})$, 
and the optimal scheduling decision at state $(0,\tau_n+1)$ is to transmit (pre-serve) the  predicted packet. 

Then, suppose the statement holds for $w=1,\dots,t$, we show that it also holds for $w=t+1$. 
From \eqref{eq:dp-aug-simple}, we have: 
\begin{eqnarray*}
%\begin{aligned}
\hspace{-.55in}&&V_{n}(0,\tau_n+t+1) \\
\hspace{-.55in}\qquad &&=\max\big\{ V_{n}(0,\tau_n+t),-\lambda +\zeta_{n}p\beta_{n}+(1-\zeta_{n})V_{n}(0,\tau_n+t)\big\},\\
\hspace{-.55in}\qquad &&=\max\big\{ -(1-p_{n})\frac{1-(1-\zeta_{n})^{t}}{\zeta_{n}}\lambda\\
\hspace{-.55in}\qquad &&\qquad\qquad +p_{n}\frac{1-(1-\zeta_{n})^{\tau_n+t}}{\zeta_{n}}(-\lambda+\zeta_{n}\beta_{n}),\\
\hspace{-.55in}\qquad &&\qquad\qquad -(1-p_{n})\frac{1-(1-\zeta_{n})^{t+1}}{\zeta_{n}}\lambda\\
\hspace{-.55in}\qquad &&\qquad\qquad +p_{n}\frac{1-(1-\zeta_{n})^{\tau_n+t+1}}{\zeta_{n}}(-\lambda+\zeta_{n}\beta_{n})\big\},
%\end{aligned} 
\end{eqnarray*}
Since $p_{n}> c_{n}$, we have  $V_{n}(0,\tau_n+t+1)=-(1-p_{n})\frac{1-(1-\zeta_{n})^{t+1}}{\zeta_{n}}\lambda+p_{n}\frac{1-(1-\zeta_{n})^{\tau_n+t+1}}{\zeta_{n}}(-\lambda+\zeta_{n}\beta_{n})$. The optimal scheduling decision at state $(0,\tau_n+t+1)$ is to transmit the predicted packet. 
This completes  the proof of (\ref{eq:value-fun-pq}). % in Theorem \ref{thm:policy-simple-impre}. 
The second case (ii) and Part (B) can similarly be shown as above. 
\end{proof}

\section{Proof of Theorem \ref{throughput-simple-p-q}} 
We prove Theorem \ref{throughput-simple-p-q} here. 
\begin{proof} (Theorem \ref{throughput-simple-p-q})   
(Part (i)) For clarity, we write $g_{I}$ and $\phi^{*}_{I}$ as explicit functions of $\bm{q}$, i.e., $g_{I}(\lambda,\bm{q})$ and $\phi^{*}_{I}(\bm{q})$. 
Since $\phi^{*}_{I}(\bm{q})=\min_{\lambda}g_{I}(\lambda,\bm{q})$, in order to prove the result, we will show that for any fixed $\lambda$, $g_{I}(\lambda,\bm{q})$ is a non-increasing function of $\bm{q}$.  

We know that  if $\zeta_{n}\beta_{n}\le\lambda$, then $\frac{\partial g_{I}}{\partial q_{n}}=0$. Otherwise, taking derivative of (\ref{eq:gI-imperfect}) with respect to $q_n$, we get (note that $\tilde{a}_n$ is also a function of $q_n$): 
\begin{align*}
\frac{\partial g_{I}}{\partial q_{n}}=&\frac{a_{n}-a_{\max}p_{n}}{(p_{n}-q_{n})^{2}}v_{n}(\lambda)\\
&+\frac{(a_{\max}p_{n}-a_{n})p_{n}}{(p_{n}-q_{n})^{2}}\cdot\frac{1-(1-\zeta_{n})^{\tau_n}}{\zeta_{n}}(-\lambda+\zeta_{n}\beta_{n}) \\
=&\frac{a_{n}-a_{\max}p_{n}}{(p_{n}-q_{n})^{2}}\left[v_{n}(\lambda)-p_{n}\frac{1-(1-\zeta_{n})^{\tau_n}}{\zeta_{n}}(-\lambda+\zeta_{n}\beta_{n})\right] 
\end{align*}

We now claim $v_{n}(\lambda)\ge p_{n}\frac{1-(1-\zeta_{n})^{\tau_n}}{\zeta_{n}}(-\lambda+\zeta_{n}\beta_{n})$. To see this, note that if $p_{n}\le c_{n}$, then $v_{n}(\lambda)=p_{n}\frac{1-(1-\zeta_{n})^{\tau_n}}{\zeta_{n}}(-\lambda+\zeta_{n}\beta_{n})$. Otherwise, if $p_{n}> c_{n}$, the optimal pre-service policy is to transmit the packet at every time-slot in the prediction window. 
But $p_{n}\frac{1-(1-\zeta_{n})^{\tau_n}}{\zeta_{n}}(-\lambda+\zeta_{n}\beta_{n})$ is the value function if the packet is not pre-served. Thus, $v_{n}(\lambda)\ge p_{n}\frac{1-(1-\zeta_{n})^{\tau_n}}{\zeta_{n}}(-\lambda+\zeta_{n}\beta_{n})$. 
Combining this with the fact that $\frac{a_{n}}{a_{\max}}\le p_{n}$, we have $\frac{\partial g_{I}}{\partial q_{n}}\le0$, which finishes the proof for part (i). 

%==the last step is not clear to me as we should need $p_n\leq c_n$ ?? == \textcolor{red}{== more to be done ==}
%\end{proof}

%\section{Proof of Theorem \ref{throughput-simple-p}} 
%We prove Theorem \ref{throughput-simple-p} here. 
%\begin{proof} (Theorem \ref{throughput-simple-p}) 
(Part (ii))  When $\bm{q}=\bm{0}$, we have: 
\begin{align*}
%\begin{aligned}
g_{I}(\lambda,\bm{p})=&\,\lambda B +\sum_{n:\zeta_{n}\beta_{n}>\lambda}\frac{a_{n}}{p_{n}}v_{n}(\lambda) \\
=&\,\lambda B+\sum_{p_{n}\le c_{n}}a_{n}\frac{1-(1-\zeta_{n})^{\tau_n}}{\zeta_{n}}(-\lambda+\zeta_{n}\beta_{n}) \\
&\,\,\quad+\sum_{p_{n}> c_{n}}\frac{a_{n}}{p_{n}}\bigg[-(1-p_{n})\frac{1-(1-\zeta_{n})^{D_{n}}}{\zeta_{n}}\lambda\\
&\,\,\quad+p_{n}\frac{1-(1-\zeta_{n})^{\tau_n+D_{n}}}{\zeta_{n}}(-\lambda+\zeta_{n}\beta_{n})\bigg],
%\end{aligned}
\end{align*}
where $n$ satisfies $\zeta_{n}\beta_{n}>\lambda$.
Similar to  the proof of Part (i), we are going to show $\frac{\partial g_{I}}{\partial p_{n}}\ge0$. 

 Apparently we have $\frac{\partial g_{I}}{\partial p_{n}}=0$ if $\zeta_{n}\beta_{n}\le\lambda$. 
Now suppose $\zeta_{n}\beta_{n}>\lambda$. If $p_{n}> c_{n}$, then
\begin{equation*}
\frac{\partial g_{I}}{\partial p_{n}}=\frac{a_{n}[1-(1-\zeta_{n})^{D_{n}}]}{p_{n}^{2}\zeta_{n}}\lambda\ge0. 
\end{equation*}
Thus, $g_{I}(\lambda,\bm{p})$ is non-decreasing in $p_{n}$ on $(c_{n},1]$. 

If $p_{n}\le c_{n}$, let $\Delta \bm{p_{n}}=(0,\dots,\Delta p_{n},0,\dots)$ with $p_{n}\ge0$. If $p_{n}+\Delta p_{n}\le c_{n}$ then $g_{I}(\lambda,\bm{p}+\Delta \bm{p_{n}})=g_{I}(\lambda,\bm{p})$. 
Otherwise, if $p_{n}+\Delta p_{n}> c_{n}$ then $g_{I}(\lambda,\bm{p}+\Delta \bm{p_{n}})\ge g_{I}(\lambda,\bm{p})$. Hence, $g_{I}(\lambda,\bm{p})$ is non-decreasing in $p_{n},\forall n$. 
 %\textcolor{red}{=this argument here is vague= =needs fix=}
% 
%If $p_{n}\beta_{n}>\lambda$ and $p_{n}> c_{n}$, then
%\begin{align*}
%g_{I}(\lambda,\tilde{\bm{p}}+\Delta \bm{p_{n}})\approx& g_{I}(\lambda,\tilde{\bm{p}})+\frac{1-(1-p_{n})^{D_{n}}}{p_{n}}\lambda\Delta p_{n} \\
%\ge& g_{I}(\lambda,\tilde{\bm{p}}),
%\end{align*}
%which gives $\frac{\partial g_{I}}{\partial p_{n}}\ge0$. 
%The proof is finished.
\end{proof}

%\section*{Acknowledgment}


\section{Proof of Theorem \ref{theorem:monotone}} \label{apd:monotone}
\begin{proof} (Theorem \ref{theorem:monotone}) 
First we show that the value function has the following properties:
\begin{enumerate}[(a)]
\item $V_{n}(0,\tau,i)<\beta_{n},\forall n,i,0\le\tau<\tau_n+D_{n}$,
\item $V_{n}(0,\tau,i)\le V_{n}(0,\tau+1,i),\forall n,i,0\le\tau<\tau_n+D_{n}$. %0\le\tau<\tau_n$ or $0\le\tau<\tau_n+D_{n}$
\end{enumerate}
Property (a) can be easily verified by induction on $\tau$. 
Now, we prove (b)  by induction on $\tau$. 
%, i.e., for each $n$,
%\begin{equation}\label{eq:value-monotonicity-general}
%V_{n}(0,\tau,i)\le V_{n}(0,\tau+1,i),\forall i,
%\end{equation}
%also by induction on $\tau$.

First, if $\tau=0$, we see that (b) holds, since $V_{n}(0,1,i)=\max_{e\in\mathcal{E}} [-\lambda e+\zeta_{n}(i,e)\beta_{n}]\ge0,\forall i$.

Now suppose (b) holds for $\tau=0,\dots,t$, we show that it holds for $\tau=t+1$. To this end, we have: 
\begin{align*}
V_{n}(0,t+1,i)=&\max_{e\in\mathcal{E}} \bigg[-\lambda e+\zeta_{n}(i,e)\beta_{n}\\
&+(1-\zeta_{n}(i,e))\sum_{j}P_{n}^{i,j}V_{n}(0,t,j)\bigg],\\
V_{n}(0,t,i)=&\max_{e\in\mathcal{E}} \bigg[-\lambda e+\zeta_{n}(i,e)\beta_{n}\\
&+(1-\zeta_{n}(i,e))\sum_{j}P_{n}^{i,j}V_{n}(0,t-1,j)\bigg].
\end{align*} 
Since $V_{n}(0,t,j)\ge V_{n}(0,t-1,j),\forall j$, for a fixed $e$, 
\begin{eqnarray*}
-\lambda e+\zeta_{n}(i,e)\beta_{n}+(1-\zeta_{n}(i,e))\sum_{j}P_{n}^{i,j}V_{n}(0,t,j)\\
\ge-\lambda e+\zeta_{n}(i,e)\beta_{n}+(1-\zeta_{n}(i,e))\sum_{j}P_{n}^{i,j}V_{n}(0,t-1,j). 
\end{eqnarray*}
Thus,  $V_{n}(0,t+1,i)\ge V_{n}(0,t,i),\forall i$.

If $\mathcal{E}$ is continuous, we have from \eqref{eq:dp} that 
\begin{eqnarray*}
\hspace{-.3in}&&e_{n}^{*}(0, \tau, i)=\argmax_{e\in\mathcal{E}} \bigg[-\lambda e+\zeta_{n}(i,e)\beta_{n}\\
\hspace{-.3in}&&\qquad\qquad\qquad\qquad+(1-\zeta_{n}(i,e))\sum_{j}P_{n}^{i,j}V_{n}(0,\tau-1,j)\bigg]. 
\end{eqnarray*}
By setting the derivative of the right hand side with respect to $e$ to $0$, we  obtain: 
\begin{equation}
\frac{d\zeta_{n}(i,e)}{de}=\frac{\lambda}{\beta_{n}-\sum_{j}P_{n}^{i,j}V_{n}(0,\tau-1,j)}.
\end{equation}
Since $\zeta_{n}(i,e)$ is a concave function of $e$, from Property (a) and (b) we know \eqref{eq:policy-monotonicity-general} holds.

On the other hand, if $\mathcal{E}$ is discrete, denote $h_{n,\tau,i}(e)=-\lambda e+\zeta_{n}(i,e)\beta_{n}+(1-\zeta_{n}(i,e))\sum_{j}P_{n}^{i,j}V_{n}(0,\tau-1,j)$. 
To prove \eqref{eq:policy-monotonicity-general}, we will show that for any consumption levels $e_{l}\le e_{u}$, if $h_{n,\tau,i}(e_{l})\ge h_{n,\tau,i}(e_{u})$, then $h_{n,\tau+1,i}(e_{l})\ge h_{n,\tau+1,i}(e_{u})$. 
From this result, we can then conclude that any $e\geq e_{n}^{*}(0, \tau, i)$ will result in $h_{n,\tau+1,i}(e)\leq h_{n,\tau+1,i}(e_{n}^{*}(0, \tau, i))$, which proves the result. 

Since $h_{n,\tau,i}(e_{l})\ge h_{n,\tau,i}(e_{u})$, we have that $\lambda(e_{u}-e_{l})-\beta_{n}(\zeta_n(i, e_u)-\zeta_n(i, e_l))+(\zeta_n(i, e_u)-\zeta_n(i, e_l))\sum_{j}P_{n}^{i,j}V_{n}(0,\tau-1,j)\ge 0$. 
%
%\textcolor{red}{shouldn't this be $\beta_{n}(\zeta_n(i, E_u)-\zeta_n(i, E_l))$??}
%
Therefore, 
\begin{align*}
&h_{n,\tau+1,i}(e_{l})- h_{n,\tau+1,i}(e_{u})\\
&=\lambda(e_{u}-e_{l})-\beta_{n}(\zeta_n(i, e_u)-\zeta_n(i, e_l))\\
&\quad+(\zeta_n(i, e_u)-\zeta_n(i, e_l))\sum_{j}P_{n}^{i,j}V_{n}(0,\tau,j) \\
&\ge\lambda(e_{u}-e_{l})-\beta_{n}(\zeta_n(i, e_u)-\zeta_n(i, e_l))\\
&\quad+(\zeta_n(i, e_u)-\zeta_n(i, e_l))\sum_{j}P_{n}^{i,j}V_{n}(0,\tau-1,j) \ge0.
\end{align*}
Here the last step follows since $V_{n}(0,\tau-1,j)\leq V_{n}(0,\tau,j)$ and $\zeta_n(i, e)$ is increasing in $e$. %\textcolor{red}{more explanation} 
This completes the proof.
\end{proof}

\section{Proof of Theorem \ref{thm:value-general-perpre}} \label{apd:value-perpre}
\begin{proof} (Theorem \ref{thm:value-general-perpre})
We can easily verify that
$$0\le V_{n}(0,\tau,i)<\beta_{n},\forall n,i,0\le\tau\le \tau_n+D_{n}.$$
by induction on $\tau$.
%Now we prove
%\begin{equation} \label{eq:value-general-perpre-upper-2}
%V_{n}(0,\tau,i)\le \tau\max_{E}\{-\lambda E+p_{n}(i_{n}^{\max},E)\beta_{n}\}, 0<\tau\le\tau_{n}+D_n.
%\end{equation} 
%
%For $\tau=1$, we know $ v_{n}(0,1,i)=\max\{-\lambda E+p_{n}(i,E)\beta_{n}\}\le\max\{-\lambda E+p_{n}(i_{n}^{\max},E)\beta_{n}\}$.
%
%Suppose \eqref{eq:value-general-perpre-upper-2} holds for $\tau=1,\dots,t$, we show that it holds for $\tau=t+1$. From \eqref{eq:dp} we have
%\begin{align*}
%&V_{n}(0,t+1,i)\\
%=&\max_{E\in\mathcal{E}} \bigg[-\lambda E+p_{n}(i,E)\beta_{n}\\
%&+(1-p_{n}(i,E))\sum_{j}P_{n}^{i,j}V_{n}(0,t,j)\bigg], \\
%\le&\max_{E\in\mathcal{E}} \bigg[-\lambda E+p_{n}(i_{n}^{\max},E)\beta_{n}\\
%&+(1-p_{n}(i_{n}^{\max},E))\sum_{j}P_{n}^{i,j}V_{n}(0,t,j)\bigg] \\
%\le&\max_{E\in\mathcal{E}} \bigg\{-\lambda E+p_{n}(i_{n}^{\max},E)\beta_{n}\\
%&+(1-p_{n}(i_{n}^{\max},E))t\max_{E}\{-\lambda E+p_{n}(i_{n}^{\max},E)\beta_{n}\}\bigg\} \\
%\le&\max_{E\in\mathcal{E}} \{-\lambda E+p_{n}(i_{n}^{\max},E)\beta_{n}\}\\
%&+\max_{E\in\mathcal{E}}(1-p_{n}(i_{n}^{\max},E))t\max_{E}\{-\lambda E+p_{n}(i_{n}^{\max},E)\beta_{n}\} \\
%=&(t+1)\max_{E}\{-\lambda E+p_{n}(i_{n}^{\max},E)\beta_{n}\}
%\end{align*}
First, we prove $V_{n}(0,\tau,i)\ge V_{n}^{l}(\tau)$ with induction. Recall that 
$$V_{n}^{l}(\tau)=\max_{e>0}\frac{1-[1-\zeta_{n}(i_{n}^{\min},e)]^{\tau}}{\zeta_{n}(i_{n}^{\min},e)}[-\lambda e+\zeta_{n}(i_{n}^{\min},e)\beta_{n}].$$
%Next we prove \eqref{eq:value-general-perpre-lower} by induction. 
For $\tau=1$, we know $ V_{n}(0,1,i)=\max_{e}\{-\lambda e+\zeta_{n}(i,e)\beta_{n}\}\ge\max_{e>0}\{-\lambda e+\zeta_{n}(i_{n}^{\min},e)\beta_{n}\}$. 

Suppose for $\tau=1,\dots,t$, $V_{n}(0,\tau,i)\ge V_{n}^{l}(\tau)$ holds. We now  show that it also holds for $t+1$. 
We have from \eqref{eq:dp} that 
\begin{align*}
&V_{n}(0,t+1,i)\\
=&\max_{e\in\mathcal{E}} \bigg\{-\lambda e+\zeta_{n}(i,e)\beta_{n}\\
&\qquad +(1-\zeta_{n}(i,e))\sum_{j}P_{n}^{i,j}V_{n}(0,t,j)\bigg\} \\
\stackrel{(a)}{\ge}&  \max_{e>0} \bigg\{-\lambda e+\zeta_{n}(i_{n}^{\min},e)\beta_{n}\\
&\qquad+(1-\zeta_{n}(i_{n}^{\min},e))\sum_{j}P_{n}^{i,j}V_{n}(0,t,j)\bigg\} \\
\stackrel{(b)}{\ge}&\max_{e>0} \bigg\{-\lambda e+\zeta_{n}(i_{n}^{\min},e)\beta_{n}+(1-\zeta_{n}(i_{n}^{\min},e))\\
&\quad\times\max_{e>0}\frac{1-[1-\zeta_{n}(i_{n}^{\min},e)]^{t}}{\zeta_{n}(i_{n}^{\min},e)}[-\lambda e+\zeta_{n}(i_{n}^{\min},e)\beta_{n}]\bigg\} \\
\stackrel{(c)}{\ge}&\max_{e>0} \bigg\{-\lambda e+\zeta_{n}(i_{n}^{\min},e)\beta_{n}+(1-\zeta_{n}(i_{n}^{\min},e))\\
&\quad\times\frac{1-[1-\zeta_{n}(i_{n}^{\min},e)]^{t}}{\zeta_{n}(i_{n}^{\min},e)}[-\lambda e+\zeta_{n}(i_{n}^{\min},e)\beta_{n}]\bigg\} \\
=&\max_{e>0}\frac{1-[1-\zeta_{n}(i_{n}^{\min},e)]^{t+1}}{\zeta_{n}(i_{n}^{\min},e)}[-\lambda e+\zeta_{n}(i_{n}^{\min},e)\beta_{n}]. 
\end{align*}
%\textcolor{red}{explain the steps}
Here (a) is due to the fact that $\zeta_{n}(i_{n}^{\min},e)\leq \zeta_{n}(i,e)$ for all $e$ and $V_{n}(0,t,j)\leq\beta_n$. (b) is due to induction and (c) is due to the $\max$ operator. 
This completes the proof of the left inequality. 

Next we prove $V_{n}(0,\tau,i)\le V_{n}^{u}(\tau)$ also by induction. Recall that  
$V_{n}^{u}(\tau)=\sum_{z=1}^{\tau}\max_{e}\{-\lambda e+\zeta_{n}(i_{n}^{\max},e)(\beta_{n}-\max\{0,V_{n}^{l}(z-1)\})\}$.
First, it is clear that  $ V_{n}(0,1,i)=\max_{e}\{-\lambda e+\zeta_{n}(i,e)\beta_{n}\}\le\max_{e}\{-\lambda e+\zeta_{n}(i_{n}^{\max},e)\beta_{n}\}$.

Suppose $V_{n}(0,\tau,i)\le V_{n}^{u}(\tau)$ holds for $\tau=1,\dots,t$, we show that it holds for $\tau=t+1$. From \eqref{eq:dp}, we have
\begin{align*}
&V_{n}(0,t+1,i)\\ 
=&\max_{e\in\mathcal{E}} \bigg\{-\lambda e+\zeta_{n}(i,e)\beta_{n}\\
&\qquad +(1-\zeta_{n}(i,e))\sum_{j}P_{n}^{i,j}V_{n}(0,t,j)\bigg\} \\
\stackrel{(d)}{\le}& \max_{e\in\mathcal{E}} \bigg\{-\lambda e+\zeta_{n}(i_{n}^{\max},e)\beta_{n}\\
&\qquad+(1-\zeta_{n}(i_{n}^{\max},e))\sum_{j}P_{n}^{i,j}V_{n}(0,t,j)\bigg\} \\
\stackrel{(e)}{\le}&\max_{e\in\mathcal{E}} \bigg\{-\lambda e+\zeta_{n}(i_{n}^{\max},e)\beta_{n}\\
&\qquad+V_{n}^{u}(t)-\zeta_{n}(i_{n}^{\max},e)\max[0,V_{n}^{l}(t)]\bigg\} \\
=&\max_{e\in\mathcal{E}} \bigg\{-\lambda e+\zeta_{n}(i_{n}^{\max},e)(\beta_{n}-\max\{0,V_{n}^{l}(t)\})\bigg\}\\
&\qquad+V_{n}^{u}(t)\\
=&\,V_{n}^{u}(t+1). 
\end{align*}
Here (d) holds similarly because $\zeta_{n}(i_{n}^{\max},e)\geq \zeta_{n}(i,e)$ for all $e$ and $V_{n}(0,t,j)\leq\beta_n$, and (e) follows from induction and the definition of $V_{n}^{u}(t)$. 
%\textcolor{red}{explain the steps } 
This finishes the proof. 
\end{proof}

%For $1\le n\le N$, let $\tilde{V}_{n}^{l}(\tau)=V_{n}^{l}(\tau)$, $\tilde{V}_{n}^{u}(\tau)=V_{n}^{u}(\tau)$ for $0<\tau\le\tau_{n}$ (recall  \eqref{eq:value-general-perpre-lower}, \eqref{eq:value-general-perpre-upper}), and define 
%\begin{equation} \label{eq:value-general-impre-lower}
%\begin{aligned}
%\tilde{V}_{n}^{l}(\tau)&=\max_{E}\bigg\{-(1-p_{n})\frac{1-[1-p_{n}(i_{n}^{\min},E)]^{\tau-\tau_{n}}}{p_{n}(i_{n}^{\min},E)}\lambda E  \\
%&+p_{n}\frac{1-[1-p_{n}(i_{n}^{\min},E)]^{\tau}}{p_{n}(i_{n}^{\min},E)}[-\lambda E+p_{n}(i_{n}^{\min},E)\beta_{n}]\bigg\},
%\end{aligned}
%\end{equation}
%\begin{equation} \label{eq:value-general-impre-upper}
%\begin{aligned}
%&\tilde{V}_{n}^{u}(\tau) \\
%=&\sum_{z=\tau_{n}+1}^{\tau}\max_{E}\{-\lambda E+p_{n}p_{n}(i_{n}^{\max},E)(\beta_{n}-\tilde{V}_{n}^{l}(z-1))\}\\
%&+p_{n}\sum_{z=1}^{\tau_{n}}\max_{E}\{-\lambda E+p_{n}(i_{n}^{\max},E)(\beta_{n}-\tilde{V}_{n}^{l}(z-1))\},
%\end{aligned}
%\end{equation}
%for $\tau_{n}<\tau\le\tau_{n}+D_{n}$.
%%Suppose the prediction window for user $n$ is $D_{n}>0$, the true positive rate is $p_{n}$, and the false negative rate is $q_{n}$. 
%Consider a predicted arrival for user  $n$, if $\lambda$ is fixed, the value function satisfies: 
%\begin{equation} \label{eq:value-general-perpre}
%\tilde{V}_{n}^{l}(\tau_{n}+w)\le V_{n}(0,\tau+w,i)\le\min\{p_{n}\beta_{n},\tilde{V}_{n}^{u}(\tau_{n}+{w})\}
%\end{equation} 
%where $0<w\leD_n$.

\section{Proof of Theorem \ref{thm:value-general-impre}} \label{apd:value-impre}
\begin{proof} (Theorem \ref{thm:value-general-impre})
First, we can easily verify that $0\le V_{n}(0,\tau_n+w,i)<p_{n}\beta_{n}$ for all $n,i$ and $0<w\le D_{n}$ by induction on $w$. 

Now we start by  proving $V_{n}(0,\tau+w,i)\ge\tilde{V}_{n}^{l}(\tau_n+w)$ for $0<w \le D_n$. 
Recall that for $0<\tau<\tau_n$, $\tilde{V}_{n}^{l}(\tau)=V_{n}^{l}(\tau)$, and for $\tau_n\le\tau\le \tau_n+D_{n}$,  $\tilde{V}_{n}^{l}(\tau)$ is defined in (\ref{eq:value-general-impre-lower}). 
% 
%=\max_{E}\bigg\{-(1-p_{n})\frac{1-[1-p_{n}(i_{n}^{\min},E)]^{\tau-\tau_{n}}}{p_{n}(i_{n}^{\min},E)}\lambda E +p_{n}\frac{1-[1-p_{n}(i_{n}^{\min},E)]^{\tau}}{p_{n}(i_{n}^{\min},E)}[-\lambda E+p_{n}(i_{n}^{\min},E)\beta_{n}]\bigg\}$. 
% 
When $w=1$, %\textcolor{red}{check steps} 
\begin{align*}
&V_{n}(0,\tau_n+1,i)\\
=&\max_{e\in\mathcal{E}} \bigg[-\lambda e+\zeta_{n}(i,e)p_{n}\beta_{n}\\
&\qquad +p_{n}(1-\zeta_{n}(i,e))\sum_{j}P_{n}^{i,j}V_{n}(0,\tau_n,j)\bigg]\\
\stackrel{(a)}{\ge}&\max_{e>0} \bigg[-\lambda e+\zeta_{n}(i_{n}^{\min},e)p_{n}\beta_{n}\\
&\qquad+p_{n}(1-\zeta_{n}(i_{n}^{\min},e))\sum_{j}P_{n}^{i,j}V_{n}(0,\tau_n,j)\bigg] \\
\stackrel{(b)}{\ge}&\max_{e>0} \bigg[-\lambda e+\zeta_{n}(i_{n}^{\min},e)p_{n}\beta_{n}\\
&\qquad+p_{n}(1-\zeta_{n}(i_{n}^{\min},e))V_{n}^{l}(\tau_n)\bigg] \\
\ge&\max_{e>0} \bigg[-(1-p_{n})\lambda e+p_{n}(-\lambda e+\zeta_{n}(i_{n}^{\min},e)\beta_{n})\\
&\qquad+p_{n}(1-\zeta_{n}(i_{n}^{\min},e))V_{n}^{l}(\tau_n)\bigg] \\
\stackrel{(c)}{\ge}&\,\tilde{V}_{n}^{l}(\tau_n+1).
\end{align*}
Here the first equality is due to $V_{n}(1,\tau_n,j)=  p_{n}\beta_{n}$ in  (\ref{eq:dp-aug-4}). Inequality (a) follows from the fact that $\zeta_{n}(i_{n}^{\min},e)\leq \zeta_{n}(i,e)$ for all $e$ and $V_{n}(0,\tau_n,j)\leq\beta_n$. (b) follows from induction and (c) follows from the definition of $\tilde{V}_{n}^{l}(\tau_n+1)$. 

Then, suppose for $w=1,\dots,t$, $V_{n}(0,\tau+w,i)\ge\tilde{V}_{n}^{l}(\tau_n+w)$, we show that it also holds for $w=t+1$.
\begin{align*}
&V_{n}(0,\tau_n+t+1,i)\\
=&\max_{e\in\mathcal{E}} \bigg[-\lambda e+\zeta_{n}(i,e)p_{n}\beta_{n}\\
&\qquad+(1-\zeta_{n}(i,e))\sum_{j}P_{n}^{i,j}V_{n}(0,\tau_n+t,j)\bigg]\\
\ge&\max_{e>0} \bigg[-\lambda e+\zeta_{n}(i_{n}^{\min},e)p_{n}\beta_{n}\\
&\qquad+(1-\zeta_{n}(i_{n}^{\min},e))\sum_{j}P_{n}^{i,j}V_{n}(0,\tau_n+t,j)\bigg] \\
\ge&\max_{e>0} \bigg[-\lambda e+\zeta_{n}(i_{n}^{\min},e)p_{n}\beta_{n}\\
&\qquad+(1-\zeta_{n}(i_{n}^{\min},e))\tilde{V}_{n}^{l}(\tau_n+t)\bigg] \\
\ge&\max_{e>0} \bigg[-(1-p_{n})\lambda e+p_{n}(-\lambda e+\zeta_{n}(i_{n}^{\min},e)\beta_{n})\\
&\qquad+(1-\zeta_{n}(i_{n}^{\min},e))\tilde{V}_{n}^{l}(\tau_n+t)\bigg] \\
\ge&\,\tilde{V}_{n}^{l}(\tau_n+t+1).
\end{align*}
Here the inequalities are derived similarly as in  the $w=1$ case. With a similar method we can prove $V_{n}(0,\tau+w,i)\ge\tilde{V}_{n}^{l}(\tau_n)$ for $0<w \le D_n$. Thus, the lower bound is proved.

Next, we prove the upper bound. 
%$V_{n}(0,\tau+w,i)\le\tilde{V}_{n}^{u}(\tau_n+w), 0<w \le D_n$. Remember for 
Recall that for $0<\tau<\tau_n$, $\tilde{V}_{n}^{u}(\tau)=V_{n}^{u}(\tau)$, and for $\tau_n\le\tau\le \tau_n+D_{n}$,  $\tilde{V}_{n}^{u}(\tau)$ is defined in (\ref{eq:value-general-impre-upper}). 
% =\sum_{z=\tau_{n}+2}^{\tau}\max_{E}\{-\lambda E+p_{n}(i_{n}^{\max},E)(p_{n}\beta_{n}-\tilde{V}_{n}^{l}(z-1))\}+\max_{E}\{-\lambda E+p_{n}p_{n}(i_{n}^{\max},E)(\beta_{n}-\tilde{V}_{n}^{l}(\tau_{n}))\}+p_{n}\sum_{z=1}^{\tau_{n}}\max_{E}\{-\lambda E+p_{n}(i_{n}^{\max},E)(\beta_{n}-\tilde{V}_{n}^{l}(z-1))\}$. If 
When $w=1$, we have: 
\begin{align*}
&V_{n}(0,\tau_n+1,i)\\
=&\max_{e\in\mathcal{E}} \bigg[-\lambda e+\zeta_{n}(i,e)p_{n}\beta_{n}\\
&\qquad+p_{n}(1-\zeta_{n}(i,e))\sum_{j}P_{n}^{i,j}V_{n}(0,\tau_n,j)\bigg] \\
\stackrel{(d)}{\le}&\max_{e\in\mathcal{E}} \bigg[-\lambda e+\zeta_{n}(i_{n}^{\max},e)p_{n}\beta_{n}\\
&\qquad+p_{n}(1-\zeta_{n}(i_{n}^{\max},e))\sum_{j}P_{n}^{i,j}V_{n}(0,\tau_n,j)\bigg] \\
\stackrel{(e)}{\le}&\max_{e\in\mathcal{E}} \bigg[-\lambda e+\zeta_{n}(i_{n}^{\max},e)p_{n}\beta_{n}\\
&\qquad+p_{n}V_{n}^{u}(\tau_n)-p_{n}\zeta_{n}(i_{n}^{\max},e)\max[0,V_{n}^{l}(\tau_n)]\bigg]\\
\stackrel{(f)}=&\max_{e\in\mathcal{E}} \bigg[-\lambda e+\zeta_{n}(i_{n}^{\max},e)(p_{n}\beta_{n}-\max\{0,\tilde{V}_{n}^{l}(\tau_n)\})\bigg]\\
&\qquad+\tilde{V}_{n}^{u}(\tau_n)\\
=&\,\tilde{V}_{n}^{u}(\tau_n+1).
\end{align*}
Here (d) follows from  $\zeta_{n}(i_{n}^{\max},e)\geq \zeta_{n}(i,e)$ for all $e$ and $V_{n}(0,\tau_n,j)\leq\beta_n$. (e) follows from induction and the definition of $V_{n}^{u}(\tau_n)$ and (f) is due to $\tilde{V}_{n}^{l}(\tau_n)=p_{n}V_{n}^{l}(\tau_n)$ and $\tilde{V}_{n}^{u}(\tau_n)=p_{n}V_{n}^{u}(\tau_n)$.

Suppose $V_{n}(0,\tau+w,i)\le\tilde{V}_{n}^{u}(\tau_n+w)$   for $w=1, \dots, t$, we show that it also holds for $\tau=t+1$. We have from \eqref{eq:dp-aug} that:  
\begin{align*}
&V_{n}(0,\tau_n+t+1,i)\\
=&\max_{e\in\mathcal{E}} \bigg[-\lambda e+\zeta_{n}(i,e)p_{n}\beta_{n}\\
&\qquad+(1-\zeta_{n}(i,e))\sum_{j}P_{n}^{i,j}V_{n}(0,\tau_n+t,j)\bigg] \\
\le&\max_{e\in\mathcal{E}} \bigg[-\lambda e+\zeta_{n}(i_{n}^{\max},e)p_{n}\beta_{n}\\
&\qquad+(1-\zeta_{n}(i_{n}^{\max},e))\sum_{j}P_{n}^{i,j}V_{n}(0,\tau_n+t,j)\bigg] \\
\le&\max_{e\in\mathcal{E}} \bigg[-\lambda e+\zeta_{n}(i_{n}^{\max},e)p_{n}\beta_{n}+\tilde{V}_{n}^{u}(\tau_n+t)\\
&\qquad-\zeta_{n}(i_{n}^{\max},e)\max\{0,\tilde{V}_{n}^{l}(\tau_n), \tilde{V}_{n}^{l}(\tau_n+t)\}\bigg] \\
=&\max_{e\in\mathcal{E}} \bigg[-\lambda e+\zeta_{n}(i_{n}^{\max},e)(p_{n}\beta_{n}\\
&\qquad-\max\{0,\tilde{V}_{n}^{l}(\tau_n), \tilde{V}_{n}^{l}(\tau_n+t)\})\bigg]+\tilde{V}_{n}^{u}(\tau_n+t)\\
=&\,\tilde{V}_{n}^{u}(\tau_n+t+1).
\end{align*}
Here the inequalities similarly follow as above. This completes the proof of the upper bound. 
\end{proof}



%\begin{thebibliography}{1}
%
%\bibitem{sigh2016}
%Rahul Singh, P. R. Kumar. Throughput optimal decentralized scheduling of multi-hop networks with end-to-end deadline constraints: Unreliable links. arXiv preprint arXiv:1606.01608, 2016.
%
%\end{thebibliography}

\end{document}